\ifx\onecol\undefined
\documentclass[journal]{IEEEtran}
\else 
\documentclass[onecolumn,draftcls,12pt]{IEEEtran}
\fi

\usepackage[utf8]{inputenc}
\usepackage{physics}
\usepackage{amsmath}
\usepackage{amssymb}
\usepackage{subfigure}
\usepackage{graphicx}
\usepackage{graphics}
\usepackage{color}
\usepackage{psfrag}
\usepackage{cite}
\usepackage{balance}
\usepackage{algorithm}
\usepackage{accents}
\usepackage{amsthm}
\usepackage{bm}
\usepackage{url}
\usepackage{algorithmic}
\usepackage[english]{babel}
\usepackage{multirow}
\usepackage{enumerate}
\usepackage{cases}
\usepackage{stfloats}
\usepackage{dsfont}
\usepackage{soul}
\usepackage{amsfonts}
\usepackage{fancyhdr}
\usepackage{hhline}
\usepackage{array}
\usepackage{booktabs}
\usepackage{framed}
\usepackage{float}
\usepackage{soul}
\usepackage{url}
\usepackage{hyperref}
\usepackage{mathtools}
\usepackage{threeparttable}
\usepackage{booktabs}

\newcommand{\mb}[1]{{  \mathbf  #1}}  

\begin{document}
\newtheorem{theorem}{Theorem}
\newtheorem{acknowledgement}[theorem]{Acknowledgement}
\newtheorem{axiom}[theorem]{Axiom}
\newtheorem{case}[theorem]{Case}
\newtheorem{claim}[theorem]{Claim}
\newtheorem{conclusion}[theorem]{Conclusion}
\newtheorem{condition}[theorem]{Condition}
\newtheorem{conjecture}[theorem]{Conjecture}
\newtheorem{criterion}[theorem]{Criterion}
\newtheorem{definition}{Definition}
\newtheorem{exercise}[theorem]{Exercise}
\newtheorem{lemma}{Lemma}
\newtheorem{corollary}{Corollary}
\newtheorem{notation}[theorem]{Notation}
\newtheorem{problem}[theorem]{Problem}
\newtheorem{proposition}{Proposition}
\newtheorem{solution}[theorem]{Solution}
\newtheorem{summary}[theorem]{Summary}
\newtheorem{assumption}{Assumption}
\newtheorem{example}{\bf Example}
\newtheorem{remark}{\bf Remark}

\newtheorem{thm}{Corollary}[section]
\renewcommand{\thethm}{\arabic{section}.\arabic{thm}}

\def\qed{$\Box$}
\def\QED{\mbox{\phantom{m}}\nolinebreak\hfill$\,\Box$}
\def\proof{\noindent{\emph{Proof:} }}
\def\poof{\noindent{\emph{Sketch of Proof:} }}
\def
\endproof{\hspace*{\fill}~\qed
\par
\endtrivlist\unskip}
\def\endproof{\hspace*{\fill}~\qed\par\endtrivlist\vskip3pt}

\def\E{\mathsf{E}}
\def\eps{\varepsilon}
\def\phi{\varphi}
\def\Lsp{{\boldsymbol L}}
\def\Bsp{{\boldsymbol B}}
\def\lsp{{\boldsymbol\ell}}
\def\Ltsp{{\Lsp^2}}
\def\Lpsp{{\Lsp^p}}
\def\Linsp{{\Lsp^{\infty}}}
\def\LtR{{\Lsp^2(\Rst)}}
\def\ltZ{{\lsp^2(\Zst)}}
\def\ltsp{{\lsp^2}}
\def\ltZt{{\lsp^2(\Zst^{2})}}
\def\ninN{{n{\in}\Nst}}
\def\oh{{\frac{1}{2}}}
\def\grass{{\cal G}}
\def\ord{{\cal O}}
\def\dist{{d_G}}
\def\conj#1{{\overline#1}}
\def\ntoinf{{n \rightarrow \infty}}
\def\toinf{{\rightarrow \infty}}
\def\tozero{{\rightarrow 0}}
\def\trace{{\operatorname{trace}}}
\def\ord{{\cal O}}
\def\UU{{\cal U}}
\def\rank{{\operatorname{rank}}}
\def\acos{{\operatorname{acos}}}

\def\SINR{\mathsf{SINR}}
\def\SNR{\mathsf{SNR}}
\def\SIR{\mathsf{SIR}}
\def\tSIR{\widetilde{\mathsf{SIR}}}
\def\Ei{\mathsf{Ei}}
\def\l{\left}
\def\r{\right}
\def\lb{\left\{}
\def\rb{\right\}}

\setcounter{page}{1}

\newcommand{\eref}[1]{(\ref{#1})}
\newcommand{\fig}[1]{Fig.\ \ref{#1}}

\def\bydef{:=}
\def\ba{{\mathbf{a}}}
\def\bb{{\mathbf{b}}}
\def\bc{{\mathbf{c}}}
\def\bd{{\mathbf{d}}}
\def\bee{{\mathbf{e}}}
\def\bff{{\mathbf{f}}}
\def\bg{{\mathbf{g}}}
\def\bh{{\mathbf{h}}}
\def\bi{{\mathbf{i}}}
\def\bj{{\mathbf{j}}}
\def\bk{{\mathbf{k}}}
\def\bl{{\mathbf{l}}}
\def\bm{{\mathbf{m}}}
\def\bn{{\mathbf{n}}}
\def\bo{{\mathbf{o}}}
\def\bp{{\mathbf{p}}}
\def\bq{{\mathbf{q}}}
\def\br{{\mathbf{r}}}
\def\bs{{\mathbf{s}}}
\def\bt{{\mathbf{t}}}
\def\bu{{\mathbf{u}}}
\def\bv{{\mathbf{v}}}
\def\bw{{\mathbf{w}}}
\def\bx{{\mathbf{x}}}
\def\by{{\mathbf{y}}}
\def\bz{{\mathbf{z}}}
\def\b0{{\mathbf{0}}}

\def\bA{{\mathbf{A}}}
\def\bB{{\mathbf{B}}}
\def\bC{{\mathbf{C}}}
\def\bD{{\mathbf{D}}}
\def\bE{{\mathbf{E}}}
\def\bF{{\mathbf{F}}}
\def\bG{{\mathbf{G}}}
\def\bH{{\mathbf{H}}}
\def\bI{{\mathbf{I}}}
\def\bJ{{\mathbf{J}}}
\def\bK{{\mathbf{K}}}
\def\bL{{\mathbf{L}}}
\def\bM{{\mathbf{M}}}
\def\bN{{\mathbf{N}}}
\def\bO{{\mathbf{O}}}
\def\bP{{\mathbf{P}}}
\def\bQ{{\mathbf{Q}}}
\def\bR{{\mathbf{R}}}
\def\bS{{\mathbf{S}}}
\def\bT{{\mathbf{T}}}
\def\bU{{\mathbf{U}}}
\def\bV{{\mathbf{V}}}
\def\bW{{\mathbf{W}}}
\def\bX{{\mathbf{X}}}
\def\bY{{\mathbf{Y}}}
\def\bZ{{\mathbf{Z}}}

\def\bxi{{\boldsymbol{\xi}}}

\def\sT{{\mathsf{T}}}
\def\sH{{\mathsf{H}}}
\def\cmp{{\text{cmp}}}
\def\cmm{{\text{cmm}}}
\def\WPT{{\text{WPT}}}
\def\lo{{\text{lo}}}
\def\gl{{\text{gl}}}

\def\tT{{\widetilde{T}}}
\def\tF{{\widetilde{F}}}
\def\tP{{\widetilde{P}}}
\def\tG{{\widetilde{G}}}
\def\tbh{{\widetilde{\mathbf{h}}}}
\def\tbg{{\widetilde{\mathbf{g}}}}

\def\mA{{\mathbb{A}}}
\def\mB{{\mathbb{B}}}
\def\mC{{\mathbb{C}}}
\def\mD{{\mathbb{D}}}
\def\mE{{\mathbb{E}}}
\def\mF{{\mathbb{F}}}
\def\mG{{\mathbb{G}}}
\def\mH{{\mathbb{H}}}
\def\mI{{\mathbb{I}}}
\def\mJ{{\mathbb{J}}}
\def\mK{{\mathbb{K}}}
\def\mL{{\mathbb{L}}}
\def\mM{{\mathbb{M}}}
\def\mN{{\mathbb{N}}}
\def\mO{{\mathbb{O}}}
\def\mP{{\mathbb{P}}}
\def\mQ{{\mathbb{Q}}}
\def\mR{{\mathbb{R}}}
\def\mS{{\mathbb{S}}}
\def\mT{{\mathbb{T}}}
\def\mU{{\mathbb{U}}}
\def\mV{{\mathbb{V}}}
\def\mW{{\mathbb{W}}}
\def\mX{{\mathbb{X}}}
\def\mY{{\mathbb{Y}}}
\def\mZ{{\mathbb{Z}}}

\def\cA{\mathcal{A}}
\def\cB{\mathcal{B}}
\def\cC{\mathcal{C}}
\def\cD{\mathcal{D}}
\def\cE{\mathcal{E}}
\def\cF{\mathcal{F}}
\def\cG{\mathcal{G}}
\def\cH{\mathcal{H}}
\def\cI{\mathcal{I}}
\def\cJ{\mathcal{J}}
\def\cK{\mathcal{K}}
\def\cL{\mathcal{L}}
\def\cM{\mathcal{M}}
\def\cN{\mathcal{N}}
\def\cO{\mathcal{O}}
\def\cP{\mathcal{P}}
\def\cQ{\mathcal{Q}}
\def\cR{\mathcal{R}}
\def\cS{\mathcal{S}}
\def\cT{\mathcal{T}}
\def\cU{\mathcal{U}}
\def\cV{\mathcal{V}}
\def\cW{\mathcal{W}}
\def\cX{\mathcal{X}}
\def\cY{\mathcal{Y}}
\def\cZ{\mathcal{Z}}
\def\cd{\mathcal{d}}
\def\Mt{M_{t}}
\def\Mr{M_{r}}
\def\O{\Omega_{M_{t}}}
\newcommand{\figref}[1]{{Fig.}~\ref{#1}}
\newcommand{\tabref}[1]{{Table}~\ref{#1}}

\newcommand{\fb}{\tx{fb}}
\newcommand{\nf}{\tx{nf}}
\newcommand{\BC}{\tx{(bc)}}
\newcommand{\MAC}{\tx{(mac)}}
\newcommand{\Pout}{p_{\mathsf{out}}}
\newcommand{\nnn}{\nn\\}
\newcommand{\FB}{\tx{FB}}
\newcommand{\TX}{\tx{TX}}
\newcommand{\RX}{\tx{RX}}
\renewcommand{\mod}{\tx{mod}}
\newcommand{\m}[1]{\mathbf{#1}}
\newcommand{\td}[1]{\tilde{#1}}
\newcommand{\sbf}[1]{\scriptsize{\textbf{#1}}}
\newcommand{\stxt}[1]{\scriptsize{\textrm{#1}}}
\newcommand{\suml}[2]{\sum\limits_{#1}^{#2}}
\newcommand{\sumlk}{\sum\limits_{k=0}^{K-1}}
\newcommand{\eqhsp}{\hspace{10 pt}}
\newcommand{\tx}[1]{\texttt{#1}}
\newcommand{\Hz}{\ \tx{Hz}}
\newcommand{\sinc}{\tx{sinc}}
\newcommand{\diag}{\mathrm{diag}}
\newcommand{\MAI}{\tx{MAI}}
\newcommand{\ISI}{\tx{ISI}}
\newcommand{\IBI}{\tx{IBI}}
\newcommand{\CN}{\tx{CN}}
\newcommand{\CP}{\tx{CP}}
\newcommand{\ZP}{\tx{ZP}}
\newcommand{\ZF}{\tx{ZF}}
\newcommand{\SP}{\tx{SP}}
\newcommand{\MMSE}{\tx{MMSE}}
\newcommand{\MINF}{\tx{MINF}}
\newcommand{\RC}{\tx{MP}}
\newcommand{\MBER}{\tx{MBER}}
\newcommand{\MSNR}{\tx{MSNR}}
\newcommand{\MCAP}{\tx{MCAP}}
\newcommand{\vol}{\tx{vol}}
\newcommand{\ah}{\hat{g}}
\newcommand{\tg}{\tilde{g}}
\newcommand{\teta}{\tilde{\eta}}
\newcommand{\heta}{\hat{\eta}}
\newcommand{\uh}{\m{\hat{s}}}
\newcommand{\eh}{\m{\hat{\eta}}}
\newcommand{\hv}{\m{h}}
\newcommand{\hh}{\m{\hat{h}}}
\newcommand{\Po}{P_{\mathrm{out}}}
\newcommand{\Poh}{\hat{P}_{\mathrm{out}}}
\newcommand{\Ph}{\hat{\gamma}}
\newcommand{\mat}[1]{\begin{matrix}#1\end{matrix}}
\newcommand{\ud}{^{\dagger}}
\newcommand{\C}{\mathcal{C}}
\newcommand{\nn}{\nonumber}
\newcommand{\nInf}{U\rightarrow \infty}

\title{Realizing Quantum Wireless Sensing Without Extra Reference Sources: 
Architecture, Algorithm, and Sensitivity Maximization 
}

\author{{Mingyao Cui, Qunsong Zeng, Zhanwei Wang, and Kaibin Huang,~\IEEEmembership{Fellow, IEEE}\vspace{-7mm}}
\thanks{The authors are with Department of Electrical and Electronic Engineering, The University of Hong Kong, Hong Kong (Emails: \{mycui, qszeng, zhanweiw, huangkb\}@eee.hku.hk). Corresponding authors: Q. Zeng; K. Huang.}}



\maketitle


\begin{abstract} 
Rydberg Atomic REceivers (RAREs) have demonstrated remarkable capabilities for radio-frequency signal measurement, enabling advanced quantum wireless sensing. Existing RARE-based sensing systems 
popularly adopt the heterodyne detection methodology, which requires an additional
reference source to serve as an \emph{atomic mixer}. However, this approach 
entails a bulky transceiver architecture and is limited in the supportable 
sensing bandwidth. To address these limitations, we propose a self-heterodyne sensing paradigm where the transmitter's self-interference naturally provides the reference signal.  We demonstrate that a self-heterodyne RARE functions as an \emph{atomic autocorrelator}, eliminating the need for external reference sources while supporting substantially wider bandwidth than conventional heterodyne methods. Next, a two-stage algorithm is devised to perform target ranging in self-heterodyne RARE systems. This algorithm is shown to closely approach the Cramér-Rao lower bound. Furthermore, we introduce the power-trajectory ($P$-trajectory) design for RAREs, which maximizes the sensing sensitivity through time-varying transmission power control. An internal noise (ITN)-limited $P$-trajectory is developed to capture the profile of the asymptotically optimal time-varying power in the presence of ITN only. 
This design is then extended to the practical $P$-trajectory by incorporating both the ITN and external noise. Numerical results validate that the proposed self-heterodyne sensing can achieve a $\sim100\:{\rm MHz}$-level bandwidth with high sensitivity, substantially surpassing existing heterodyne counterparts and paving the way for high-resolution quantum wireless sensing.
\end{abstract} 

\vspace{-2mm}
\section{Introduction}

The precise measurement of radio-frequency (RF) signals is fundamental to the digital age, underpinning critical applications in wireless communication, remote sensing, and e-health.
Rydberg Atomic Receivers (RAREs), a concept originating from quantum sensing, have emerged as a promising technology for high-precision wireless detection by leveraging the quantum properties of Rydberg atoms~\cite{RydMag_Liu2023, RydbMag_Cui2024, RydMag_Fancher2021}. 
These atoms are highly excited atoms in which one or more electrons have transitioned from their ground-state energy level to a higher energy state. 
Due to their large electric dipole moments, Rydberg atoms can strongly interact with incident RF signals, triggering electron transitions between resonant energy levels~\cite{RydMag_Fancher2021}.
Capitalizing on these transitions, RAREs can capture the amplitude, frequency, phase, and polarization of RF signals with unparalleled precision.
This capability positions RAREs as potential successors to, or collaborators with, conventional RF receivers for next-generation wireless systems~\cite{AtomicMIMO_Cui2025, RydChen_Gong2025, zhang_rydberg_2024, QuanSense_Zhang2023,Shawei2025}.

The principal advantages of RAREs include exceptional sensitivity and an ultra-wide range of detectable frequency.  
The large electric dipole moments of Rydberg atoms can significantly amplify the incoming RF signals~\cite{cai_sensitivity_2023}, enabling sensitivities that surpass those of classical receivers. 
Theoretical analysis based on the standard quantum limit (SQL) predicts a fundamental sensitivity for RAREs on the order of $10\:{\rm pV/cm/\sqrt{Hz}}$~\cite{fan_atom_2015}, which is two orders of magnitude higher than the approximate limit of traditional technology, $1\:{\rm nV/cm/\sqrt{Hz}}$~\cite{Resonant_San2024}. 
Additionally, a RARE can detect signals across a vast frequency spectrum, spanning from Megahertz (MHz) to Terahertz (THz) bands~\cite{zhang_rydberg_2024,Multiband_Cui2025}. 
This wide tunability is achieved by exploiting the electron transitions among different energy levels. For example, experiments in \cite{RydMultiband_Meyer2023} have demonstrated the simultaneous detection of 5 distant frequencies spanning from 1.7 GHz to 116 GHz using a single RARE.
The detectable frequency range of RAREs can be further expanded by utilizing various quantum techniques to split atomic energy levels, such as the Zeeman effect and AC Stark shift~\cite{shi_tunable_2023, hu_continuously_2022,RydAM_Zhen2019}.


Research in RARE-enabled wireless sensing has progressed significantly. 
Initial efforts demonstrated the detection of RF signal amplitude and frequency using electromagnetically-induced-transparency (EIT) spectroscopy~\cite{RydPhase_Meyer2018, RydAMFM_Anderson2021, jia_properties_2024}. A major advancement came with the development of \emph{heterodyne sensing}, which unlocked the capability to measure signal phase~\cite{RydPhase_Simons2019, RydNP_Jing2020}. 
This technique introduces an additional reference source to transmit known reference signals to the RARE. The reference signal inside the atomic vapor cell acts as an \emph{atomic mixer} that down-converts the target RF signal to an intermediate frequency, thereby facilitating phase recovery~\cite{gao_rydberg_2025}. Owing to its high sensitivity and simplified detection procedure, heterodyne sensing has become the mainstream approach for RARE-based wireless sensing.  
It was applied to subtle vibration monitoring, where RAREs have been shown to improve sensing accuracy by at least an order of magnitude compared to traditional Wi-Fi sensing~\cite{QuanSense_Zhang2023}.
Another key application of heterodyne sensing is angle-of-arrival (AoA) detection using RARE arrays \cite{RydAOA_Robinson2021,CE_Ryd2025}. 
Compatible signal processing techniques, like multiple signal classification 
(MUSIC) and rotational invariant techniques (ESPRIT), have been adapted to enhance AoA estimation accuracy~\cite{QWC_kim2025,RydESPRIT_Gong2025}. Further innovations continue to expand the RARE toolkit. 
{\color{black} A homodyne sensing technology that employs zero intermediate frequency between reference and sensing signals was developed in \cite{QWS_chen2025} to conduct radar ranging. Meanwhile, the concept of integrated sensing and communications based on  RAREs was introduced in~\cite{QWSMag_chen2025}.  
The recent study~\cite{Multiband_kim2025} also exploits the vast energy-level resources of Rydberg atoms to facilitate multi-band sensing.
}

Despite these advancements, existing heterodyne-sensing systems are hindered by two critical limitations that impede their deployment in practical wireless sensing systems. 
First, the requirement for a dedicated reference source results in a bulky and costly hardware architecture, increased power consumption, and the need for complex self-interference cancellation between the transmitter and the receiver.
{\color{black} Second, despite the wide detectable frequency range of RAREs, their supportable \emph{instantaneous bandwidth} is inherently narrow.
This bandwidth characterizes the minimum signal pulse duration that can trigger a measurable response on the receiver~\cite{IF_IB}. It is fundamentally limited by the finite response time of the quantum state transition,  typically resulting in a bandwidth of less than~$10\:{\rm MHz}$~\cite{yang_highly_2024, RydPhase_Meyer2018}.}
Consequently, this characteristic restricts the maximum bandwidth of the sensing signal and ultimately caps the achievable sensing resolution. Together, these limitations present a significant challenge for quantum heterodyne sensing in supporting future high-resolution wireless applications.

To overcome these limitations, we propose a novel sensing paradigm termed \emph{self-heterodyne sensing}. The core idea is to \emph{leverage the self-interference from the transmitted sensing signal as the inherent reference}, thereby eliminating the need for a dedicated reference source. 
The main contributions of our work are summarized below.
\begin{itemize}
    \item \textbf{Architecture and property of self-heterodyne sensing}: 
    We introduce a new RARE architecture that removes the external reference source and the associated transmitter-receiver decoupling circuitry.
    The key innovation is the use of a single transmitted linear-frequency modulated (LFM) signal that acts dually as both the sensing signal and its own reference---hence giving the name 
    ``self-heterodyne sensing". 
    Our analysis reveals a fundamental property: the self-heterodyne RARE functions as an \emph{atomic autocorrelator}. In this regime, the received signal constitutes the \emph{autocorrelation of the transmitted LFM waveform at different time delays}. This property effectively maps the target's range to a distinct frequency at the receiver, enabling high-precision quantum range sensing via frequency estimation. The second critical property is the inherent control over the received signal's instantaneous bandwidth. By using identical LFM waveforms for sensing and reference, the received signal's bandwidth can be precisely tailored to fit within the narrow detectable bandwidth of a RARE. This is a key advantage over conventional heterodyne sensing, which employs single-tone references and lacks such a control mechanism.

    \item \textbf{Algorithmic design for self-heterodyne sensing}:  We formulate the range sensing problem as a nonlinear least-square (NLS) optimization. To solve this efficiently, we develop a two-stage algorithm. The first stage performs coarse estimation by exploiting the slow variation of the signal amplitude compared to its phase. It directly extracts the target range from the principal frequency component of the received signal. The second stage refines this initial estimate using Newton's method, which explicitly accounts for the time-varying signal amplitude induced by the dynamic sensitivity of Rydberg atoms. Numerical results demonstrate that the estimation accuracy of the proposed algorithm approaches the Cramér-Rao lower bound (CRLB).

    \item \textbf{Sensitivity maximization for self-heterodyne sensing}: 
    We propose a novel \emph{power-trajectory} ($P$-trajectory) optimization technique, aimed at maximizing the sensitivity of 
    RARE by manipulating the time-varying transmission 
    power. 
    This strategy is inspired by the finding that the gain of a self-heterodyne RARE exhibits a complex, non-monotonic relationship with respect to (w.r.t) the transmission power, rendering conventional \emph{fixed power} transmission sub-optimal. 
    Our design proceeds in two steps: First, we design an internal noise (ITN)-limited  $P$-trajectory by 
    identifying the asymptotically optimal power profile in the presence of ITN only. 
    We prove that the optimal power should vary \emph{linearly} with the instantaneous carrier frequency in the off-resonant regime, but remain \emph{constant} in the resonant regime. Second, we refine this ITN-limited trajectory to account for both internal and external noise (ETN) using the Primal-Dual Subgradient (PDS) method. Numerical results confirm that the optimized $P$-trajectory significantly enhances the system's sensing accuracy.

\end{itemize}

The remainder of the paper is organized as follows. The classic and quantum wireless sensing systems are reviewed in Section~\ref{sec:2}. The architecture and property of self-heterodyne sensing are introduced in Section~\ref{sec:3}, followed by the two-stage self-heterodyne sensing algorithm in Section~\ref{sec:4}. Section~\ref{sec:5} elaborates on $P$-trajectory design for sensitivity maximization. Numerical experiments are conducted in Section~\ref{sec:6}, and conclusions are drawn in Section~\ref{sec:7}.

\vspace{-1mm}
\section{Review of Existing Wireless Sensing Paradigms}  \label{sec:2}
{This section provides preliminaries useful for understanding the proposed 
designs. {\color{black} Throughout the paper, we focus on the quasi-monostatic sensing scenario, where the transmitter and receiver are closely deployed but not perfectly co-located due to different transceiver architectures. The transmitter-receiver separation is intentionally kept small (20$\sim$100~cm), which is negligible compared to the range of sensing targets (100~m$\sim$100~km).}
We start with an introduction to the basic principles of classic 
wireless sensing for the purpose of comparison. We then introduce the electron 
	transition model and the heterodyne 
	sensing technique utilized in existing quantum wireless sensing systems. 
    }

\begin{figure}
    \centering
    \includegraphics[width=3.5in]{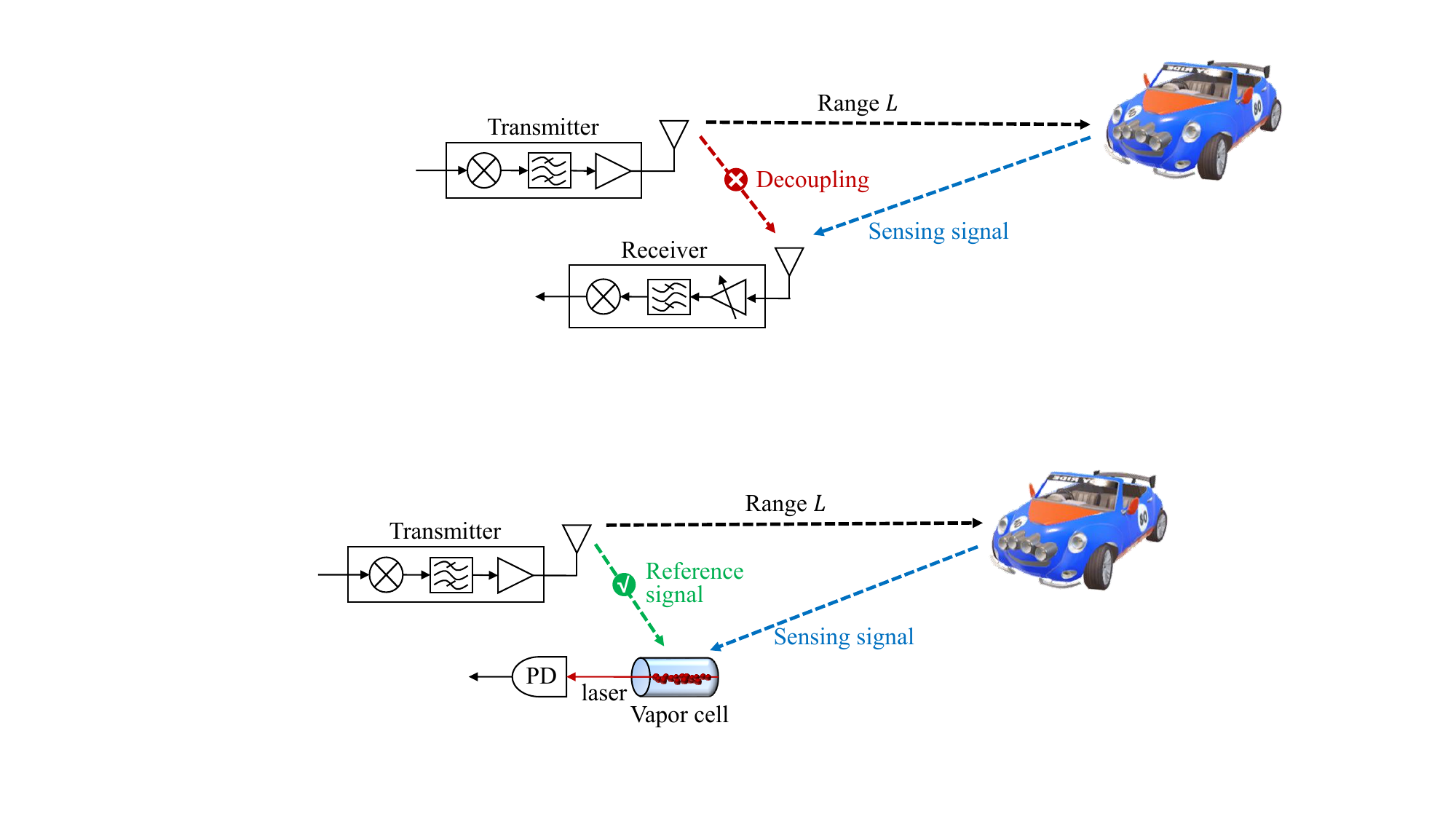}
    \vspace*{-0.5em}
    \caption{Classic wireless sensing system.} 
	\vspace*{-1em}
	\label{img:ClassicSensing}
\end{figure}

\vspace{-1mm}
\subsection{Review of Classic Wireless Sensing}\label{sec:2.1}

The architecture of the classic quasi-monostatic sensing system is depicted in Fig.~\ref{img:ClassicSensing}.
The transmitter employs the LFM waveform to sense the target 
range~\cite{Radar}, expressed as:
\begin{align}\label{eq:LFM}
    s_0(t) = \sqrt{P_{\rm tx }}e^{j\vartheta_0(t)},\:\:0\le t \le T
\end{align}
where the phase modulation function is given as $\vartheta_0(t) = \frac{1}{2}\alpha t^2 + \omega_0 t$. 
The parameters $P_{\rm tx}$, $T$, $\omega_0$, and $\alpha$ represent the transmission power, symbol duration, start frequency, and sweep slope of the LFM waveform.
The sweep bandwidth of the LFM waveform is defined as $B = \frac{\alpha T}{2\pi}$.

Let $L > 0$ denote the target range to be estimated,  which yields a roundtrip propagation delay of $\tau = \frac{2L}{c}$, with $c$ representing the speed of light. 
The roundtrip channel fading that converts the transmitted signal into the incident Poynting vector is accordingly modeled as $h = \sqrt{\frac{A_{\rm c} }{16\pi^2 L^4}}$. where $\lambda_0 = \frac{2\pi c}{\omega_0}$ is the wavelength corresponding to the start frequency $\omega_0$, $A_{\rm c}$ the cross-section area of the target, and $G_{\rm tx}$ the transmit antenna gain~\cite{malanowski_analysis_2014}. 
At the classic receiver, the received baseband complex signal is expressed as:
\begin{align}\label{eq:y0}
    y_0(t) = \sqrt{P_{\rm tx}G_0A_{\rm e}}he^{j\vartheta_0(t - \tau)} + n_0(t), 
\end{align}
where $G_0 = G_{\rm tx}G_{\rm rx} G_{\rm lna}$ encompasses the gains of the transmit antenna, $G_{\rm tx}$, receive antenna, $G_{\rm rx}$, and the low-noise amplifier (LNA), $G_{\rm rec}$. $A_{\rm e}$ characterizes the antenna aperture. 
The term $n_0(t)$ denotes zero-mean Gaussian noise with the property: $\mathrm{E}(n_0(t)n_0^*(t'))=\sigma_0^2\delta(t - t')$, whose spectrum density is $\sigma_0^2=k_BT_{\rm E}$ with $k_B$ being the Boltzmann constant and $T_{\rm E}$ the noise temperature. Accordingly, the signal-to-noise ratio (SNR) of $y_0(t)$ is given as 
\begin{align}\label{eq:SNR0}
	{\rm SNR}_0 = \frac{P_{\rm tx}G_0A_eh^2}{\sigma_0^2}.
\end{align}

In the literature, the target range $L$ can be estimated from the received signal $y_0(t)$ using a variety of signal processing techniques, such as time-frequency analysis~\cite{boashash_estimating2_1992} and NLS methods~\cite{nandi_asymptotic_2004}. 
It is worth noting that the symbol duration, $T$, is typically much longer than the roundtrip propagation time, $\tau$, (e.g., $T = 1 {\rm ms}$ and $\tau = 1\mu {\rm s}$ for $R= 150$m). 
This disparity implies that the receiver must operate in a ``listen-while-transmit'' mode, necessitating proper signal decoupling between the transmitter and receiver to mitigate the strong self-interference from the transmitted signal.



\subsection{Review of Quantum Wireless Sensing}
\begin{figure}
    \centering
    \includegraphics[width=3.5in]{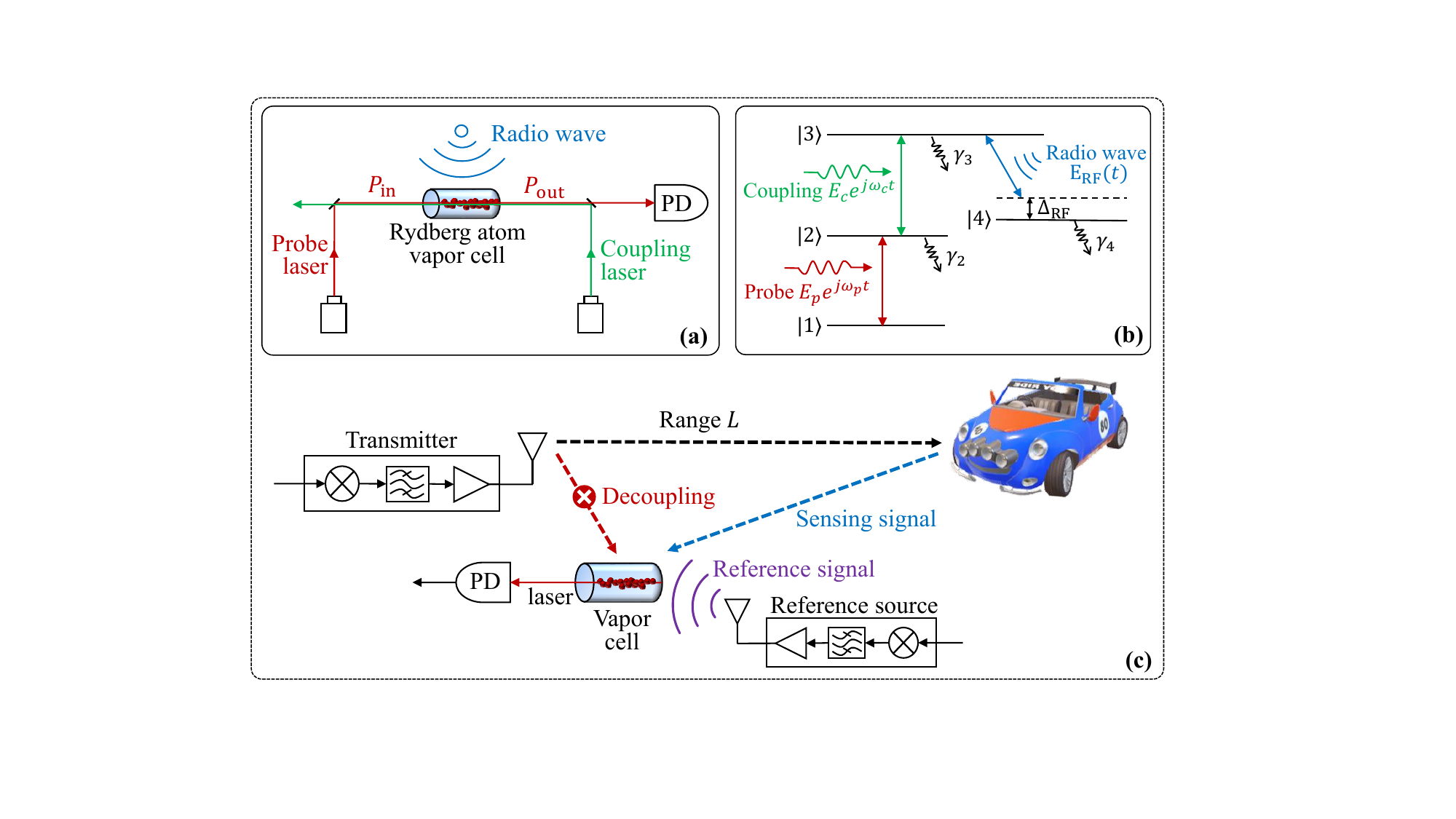}
    \vspace*{-0.5em}
    \caption{(a) Structure of a RARE; (b) 
    the 
    four-level quantum system; (c) the quantum heterodyne sensing system. } 
	\vspace*{-1em}
	\label{img:HeterodyneSensing}
\end{figure}

The quantum wireless sensing system uses a RARE in place of traditional RF 
receivers to capture incident RF signals. 

\subsubsection{Electron transition model}
We first review the electron transition model of RARE. 
Fig.~\ref{img:HeterodyneSensing}(a) presents a typical RARE structure. A probe laser and a coupling laser propagate in opposite directions through an alkali atom vapor cell to prepare Rydberg atoms. The incident RF signal changes the quantum states of Rydberg atoms.
This interaction results in variations in the power loss of the probe laser as it passes through the vapor cell.
By monitoring these variations using a photodetector (PD), the RF signal can be recovered, thereby enabling target sensing.

\emph{i) Quantum response of the four-level system:} 
The quantum response of each Rydberg atom can be characterized by the four-level quantum system in Fig.~\ref{img:HeterodyneSensing}(b). It consists of a ground state $\ket{1}$, a lowly-excited state $\ket{2}$, and two Rydberg states $\ket{3}$ and $\ket{4}$.
These states are interconnected by the probe laser, the coupling laser, and the incident RF signal, respectively. 
Let $\omega_{ij}$ denote the transition frequency of $\ket{i} \rightarrow \ket{j}$, where $i,j\in\{1,2,3,4\}$. 
The probe laser is characterized by a Rabi frequency, $\Omega_{\rm p}$, a carrier frequency, $\omega_{\rm p}$, and a frequency detuning, $\Delta_{\rm p}$. 
The Rabi frequency $\Omega_{\rm p}$ quantifies the strength of the interaction between the probe laser and the atomic transition. 
The frequency detuning, $\Delta_{\rm p}$, represents the deviation of the carrier frequency, $\omega_{\rm p}$, from the transition frequency $\omega_{12}$, i.e., $\Delta_{\rm p} = \omega_{\rm p} - \omega_{12}$. 
Likewise, we introduce $\{\Omega_{\rm c}, \omega_{\rm c}, \Delta_{\rm c} = \omega_{\rm c} - \omega_{23}\}$ for the coupling laser and $\{\Omega_{\rm RF}, \omega_{\rm RF}, \Delta_{\rm RF} = \omega_{\rm RF} - \omega_{34}\}$ for the RF signal. 
In this study, both the probe and coupling lasers are set resonant  with their respective electron transitions, such that $\Delta_{\rm p} = \Delta_{\rm c} = 0$.

The quantum state of the considered four-level system is a density matrix
$\boldsymbol{\rho}\in\mathbb{C}^{4\times 4}$. The dynamics of $\boldsymbol{\rho}$ is governed by the Lindblad master equation \cite{RydNP_Jing2020}: 
\begin{align}\label{eq:lindblad}
    \frac{\partial \boldsymbol{\rho}}{\partial t} = -\frac{j}{\hbar}[{\mb{H}}, \boldsymbol{\rho}] + \mathcal{L},
\end{align}
where $[\mb{H},\boldsymbol{\rho}]=\mb{H}\boldsymbol{\rho}-\boldsymbol{\rho}\mb{H}$ represents the commutator of the matrix $\mb{H}$ with the matrix $\boldsymbol{\rho}$. The matrix, ${\mb{H}}$, refers to the system Hamiltonian operator
\begin{align}\small\label{eq:hamiltonian}
    {\mb{H}} = \frac{\hbar}{2}\left[
    \begin{array}{cccc}
         0& {\Omega_{\rm p}} &0 &0   \\
         {\Omega_{\rm p}}& 0 & {\Omega_{\rm c}} &0 \\
         0 &{\Omega_{\rm c}} &0 & {\Omega_{\rm RF}} \\
         0 & 0 &  {\Omega_{\rm RF}} & -2\Delta_{\rm RF}
    \end{array}
    \right],
\end{align}
where $\hbar$ denotes the reduced Plank constant.
The operator $\mathcal{L}$ indicates the relaxation of the system,  
\begin{small}
\begin{align}\label{eq:decay}
    \mathcal{L} = \left[
    \begin{array}{cccc}
         \gamma_2\rho_{22} + \gamma_4 \rho_{44}& -\frac{\gamma_2}{2}\rho_{12} & -\frac{\gamma_3}{2}\rho_{13} & -\frac{\gamma_4}{2}\rho_{14}   \\
         -\frac{\gamma_2 }{2}\rho_{21}& 
         \gamma_3\rho_{33} - \gamma_2\rho_{22}
         & -\frac{\gamma_{23}}{2}\rho_{23} &
         -\frac{\gamma_{24} }{2}\rho_{24}\\
         -\frac{\gamma_3 }{2}\rho_{31} &
         -\frac{\gamma_{23}}{2}\rho_{32} &
         -\gamma_3\rho_{33} &
         -\frac{\gamma_{34}}{2}\rho_{34} \\
         -\frac{\gamma_4}{2}\rho_{41} & -\frac{\gamma_{24}}{2}\rho_{42} & - \frac{\gamma_{34}}{2}  \rho_{43} & -\gamma_4 \rho_{44}
    \end{array}
    \right] \notag
\end{align}
\end{small}
Here, $\gamma_{ij} = (\gamma_i + \gamma_j)$, where $\gamma_i\: (i = 2, 3, 4)$ is the decay rate.
We are interested in the steady-state (when $\frac{\partial \boldsymbol{\rho}}{\partial t} = 0$) solution of $\rho_{12}$ as it is associated with the probe laser to be measured, which is proven to be~\cite{RydModel_RuiNi2023}
\begin{align}
    \rho_{12} = \frac{A_1\Omega_{\rm RF}^2\Delta_{\rm RF}^2 + 
    jB_1\Omega_{\rm RF}^4}{C_1\Omega_{\rm RF}^4 + C_2\Omega_{\rm RF}^2 + 
    C_3\Delta_{\rm RF}^2}.
\end{align}
Here, $A_1 = 2\Omega_{\rm p}\Omega_{\rm c}^2$, $B_1 = \gamma_2\Omega_{\rm p}$, $C_1 = 2\Omega_{\rm p}^2 + \gamma_2^2$, $C_2 = 2\Omega_{\rm p}^2(\Omega_{\rm c}^2 + \Omega_{\rm p}^2)$, and $C_3 = 4(\Omega_{\rm c}^2 + \Omega_{\rm p}^2)^2$. 

\emph{ii) Optical readout:} 
According to the adiabatic approximation, the power of the probe laser coming out of the vapor cell is determined by the imaginary part of $\rho_{12}$~\cite{RydNP_Jing2020}:
\begin{align}
    P_{\rm out} = P_{\rm in}\exp(-{C_0\Im{\rho_{12}}}),
\end{align} where $P_{\rm in}$ denotes the input power of probe laser. The constant $C_0$ is given as $C_0 \overset{\Delta}{=} \frac{2N_0\mu_{12}^2k_pd}{\epsilon_0\hbar\Omega_{\rm p}}$, where $N_0$ represents the total density of atoms, $\mu_{12}$ the transition dipole moment of $\ket{1}\rightarrow\ket{2}$, $\epsilon_0$ the vacuum permittivity, $d$ the length of vapor cell, and $k_p$ the wavenumber of probe laser. {\color{black}The PD captures the probe laser and converts it into photocurrent:
$I_{\rm out} = \frac{q\eta}{\hbar \omega_{\rm p}}P_{\rm out}
$, where $\eta$ is the quantum efficiency of the PD, and $q$ is the elementary charge. 
An output voltage is produced from the photocurrent via a trans-impedance amplifier (TIA) for measurement 
\begin{align}
    V_{\rm out} = R_{\rm T}I_{\rm out} \overset{\Delta}{=} V_{\rm in}\exp(-C_0\Im{\rho_{12}}),  
\end{align}
where $R_{\rm T}$ is the load impedance. 
 The equivalent input voltage, $V_{\rm in} \overset{\Delta}{=} \frac{R_{\rm T}q\eta}{\hbar \omega_{\rm p}}P_{\rm in}$, is introduced to ease expression, which characterizes the voltage induced by the input power of probe laser.
For brevity, we introduce the function 
\begin{align}\label{eq:Pi}
    \Pi(\Omega, \Delta) \overset{\Delta}{=} V_{\rm in} \exp\left\{
    -\frac{B_1C_0\Omega^4}{C_1\Omega^4 + C_2\Omega^2 + C_3\Delta^2}
    \right\}, 
\end{align}
with $\Omega \in [0, +\infty)$ and $\Delta \in \mathbb{R}$, to rewrite the output voltage $V_{\rm out}$ as 
${\Pi}(\Omega_{\rm RF}, \Delta_{\rm RF})$, and define the partial derivative of $\Pi(\Omega, \Delta)$ w.r.t $\Omega$ as 
\begin{align}\small\label{eq:Upsilon}
    &\Upsilon(\Omega, \Delta) \overset{\Delta}{=} \frac{\partial \Pi(\Omega, \Delta)}{\partial 
\Omega} =  -2V_{\rm in} B_1C_0\cdot \\&\exp\left\{
    -\frac{B_1C_0\Omega^4}{C_1\Omega^4 + C_2\Omega^2 + C_3\Delta^2}
    \right\} \frac{\Omega^3(C_2\Omega^2 + 2C_3\Delta^2)}{(C_1\Omega^4 + C_2\Omega^2 + C_3\Delta^2)^2}.\notag
\end{align} 
} 


\subsubsection{Heterodyne sensing}
{\color{black} To accurately extract target information from the RF signal, the heterodyne sensing technique is commonly employed~\cite{RydNP_Jing2020, RydPhase_Simons2019} as it is more suitable for weak-signal detection than the EIT-based detection~\cite{AtomicMIMO_Cui2025}. }
This technique,  as illustrated in 
Fig.~\ref{img:HeterodyneSensing}(c), introduces an additional reference source near the RARE. 
The reference source emits a reference signal to the RARE, which interferes with the sensing 
signal reflected by the target to form the incident RF signal. Notably, the reference signal needs to be much stronger than the sensing signal and its waveform should be known to the RARE. 

\emph{i) Incident electric field model:} 
In the literature, the reference field is usually a sinusoidal signal, used to lock the transition state of RARE. 
Let $\omega_{\rm r}$ be the carrier frequency of the reference signal. 
Then, the electric field component [unit:~$\rm{V/m}$] of the incident reference signal is expressed as\footnote{\color{black} We employ a point receiver model to characterize the incident field, which neglects the spatial variation of electric field across the vapor cell. This simplification is justified because the cell length we considered ($L = 2$~cm) is substantially smaller than the field's wavelength ($d \approx 10\:{\rm cm}$), ensuring minimal phase and amplitude variation across the vapor volume.} 
\begin{align}
    E_{\rm r}(t) = E_{\rm r}e^{j(\omega_{\rm r} t - \omega_{\rm r} \tau')},
\end{align}
where $E_{\rm r} > 0$ represents the field strength and $\tau'$ the reference-to-RARE propagation delay. 
Therein, the 
frequency detuning of $E_{\rm r}(t)$ is calculated as 
$\Delta_{\rm r} = \omega_{\rm r} - \omega_{34}$. The Rabi frequency of $E_{\rm r}(t)$, is determined by the product of transition dipole moment and field strength: 
$\Omega_{\rm r} = \frac{\mu_{34}}{\hbar}|E_{\rm r}e^{j(  \omega_{\rm r}t -   \omega_{\rm r}\tau')}| = 
		\frac{\mu_{34}}{\hbar}E_{\rm r}$, 
		where $\mu_{34}$ stands for the 
dipole moment of the transition $\ket{3}\rightarrow\ket{4}$. 

As for the sensing signal, its waveform can be properly configured to deal with different sensing tasks. For this purpose, we denote its time-varying phase modulation function in a general form as $\vartheta_{\rm s}(t)$. 
Thereby, the electric field component of the received sensing signal is modeled as  
\begin{align}
	E_{\rm s}(t) = E_{\rm s} e^{j\vartheta_{\rm s}(t-\tau)},
\end{align}
where $E_{\rm s} > 0$ denotes the field strength and $\tau$ the roundtrip propagation 
time to be estimated. Likewise, the Rabi frequency of $E_{\rm s}(t)$  
is 
 $\Omega_{\rm s}= \frac{\mu_{34}}{\hbar}|{E_{\rm s} e^{j\vartheta_{\rm s}(t-\tau)}}| = \frac{\mu_{34}}{\hbar}E_{\rm s} $.  
 
As a result, the overall incident RF field is the superposition of the reference and 
sensing signals: 
\begin{align}\label{eq:hs}
E_{\rm RF}(t) = E_{\rm r} e^{j(  \omega_{\rm r}t -   \omega_{\rm r}\tau')}+ E_{\rm s} e^{j\vartheta_{\rm s}(t - \tau)}. 
\end{align}
The time-varying phase difference between $E_{\rm s}(t)$ and $E_{\rm r}(t)$ leads to a time-varying amplitude of the composite field, $E_{\rm RF}(t)$. This further results in a time-varying Rabi frequency associated with the overall RF field: 
%
%
\begin{align}\label{eq:Omega_RF}
    \Omega_{\rm RF}(t) &=\frac{\mu_{34}}{\hbar} \left| {E_{\rm r}}e^{j(  \omega_{\rm r}t - 
      \omega_{\rm r}\tau')}+ {E_{\rm 
    		s}}e^{j\vartheta_{\rm s}(t - \tau)}\right| \notag \\
    &=|\Omega_{\rm r} + \Omega_{\rm s}e^{j\vartheta_{\rm sr}(t; \tau, \tau')}|.
\end{align}
Here, $\vartheta_{\rm sr}(t; \tau, \tau')\overset{\Delta}{=}\vartheta_{\rm s}(t-\tau) - 
  \omega_{\rm r}t  
+   \omega_{\rm r}\tau'$ characterizes the phase difference function. 
Moreover, as the reference signal is much stronger than the sensing signal, $E_r \gg E_s$ ($\Omega_r \gg \Omega_s$), the frequency detuning, $\Delta_{\rm RF}$, of the overall field is 
largely determined by the reference component: 
\begin{align}\label{eq:Delta_RF}
	\Delta_{\rm RF} = 
	\Delta_r = \omega_{\rm r} - \omega_{34}.
\end{align}

\emph{ii) Output voltage and noise:} 
By substituting \eqref{eq:Omega_RF} and \eqref{eq:Delta_RF} into the output voltage, 
$V_{\rm out} = {\Pi}(\Omega_{\rm RF}, \Delta_{\rm RF})$, and taking the noise into account, the measured signal is given as
\begin{align}\label{eq:signal_model}
    &y(t) =
    \Pi(|\Omega_{\rm r} + \Omega_{\rm s}e^{j\vartheta_{\rm sr}(t; \tau, \tau')}|, \Delta_r) + n(t) 
     \\
    &\overset{(a)}{\approx} \underbrace{\Pi(\Omega_{\rm r}, \Delta_r)}_{\rm bias} + 
    \underbrace{\frac{\mu_{34}}{\hbar}\Upsilon(\Omega_{\rm r}, 
    \Delta_r)}_{\rm intrinsic\:gain}E_s\underbrace{\cos(\vartheta_{\rm sr}(t; 
    \tau, 
    \tau'))}_{\rm oscillation\:term} 
    + \underbrace{n(t)}_{\rm noise}, \notag
    \end{align}
where linearization (a) comes from the first-order Taylor expansion 
of $\Pi(\Omega_{\rm RF}, \Delta_r)$ when $\Omega_r \gg \Omega_s$. The terms $\Pi(\Omega_{\rm r}, \Delta_r)$ and $\frac{\mu_{34}}{\hbar}\Upsilon(\Omega_{\rm r}, 
    \Delta_r)$ represent the direct-current bias of the measured voltage and the intrinsic gain of RARE, respectively. 
{\color{black} 
The noise $n(t)$, with autocorrelation $\mathrm{E}(n(t')n(t))=\sigma^2\delta(t - t')$, originates from both external and internal sources. The external noise (ETN) arises from the blackbody radiation and vacuum fluctuation. Its field intensity is given by $\langle E_I^2 \rangle = \frac{\hbar\omega_0^3}{\pi \epsilon_0 c^3} (2n_{\rm th} + 1)$, where $n_{\rm th} = 1/\left(e^{\hbar\omega_0/k_BT_E} - 1\right)$ represents the Bose-Einstein distribution and $\epsilon_0$ the dielectric constant~\cite{Multiband_Cui2025}. 
The ETN corrupts the sensing signal in free space and is thus amplified by the intrinsic gain of the RARE system. 
The internal noise (ITN) is dominated by the photon shot noise of the PD, which scales proportionally with the direct-current bias. 
As a result, the spectral density of $n(t)$ is given by the superposition of the ETN and ITN contributions, $\sigma_{\rm etn}^2$ and $\sigma_{\rm itn}^2$ respectively: 
\begin{align}\label{eq:noise_hs}
    \sigma^2 = \underbrace{\frac{\mu_{34}^2}{\hbar^2}\Upsilon^2(\Omega_r,\Delta_r) \langle E_I^2 \rangle}_{\sigma_{\rm etn}^2}
+\underbrace{q R_{\rm T}\Pi(\Omega_r,\Delta_r)}_{{\sigma_{\rm itn}^2}}.
\end{align}
}

\emph{iii) Wireless sensing principle:} 
The signal model in \eqref{eq:signal_model} reveals that the reference signal 
acts as a down-converter with a down-conversion frequency of 
$\omega_{\rm r}$. 
Henceforth, a heterodyne-sensing RARE is also regarded as an atomic mixer~\cite{RydPhase_Simons2019}.
By manipulating the waveform of sensing signal, $e^{j\vartheta_{\rm s}(t)}$, the 
RARE is able to conduct diverse sensing tasks. 
Below are two examples illustrating this capability.

\begin{itemize}
    \item \emph{Example 1} ($\vartheta_{\rm s}(t) = \omega_{\rm s} t$): When the sensing signal has a single tone, the received waveform is $y(t)\propto \cos((\omega_{\rm s} - \omega_{\rm r}) t - 
\omega_{\rm s}\tau + 
  \omega_{\rm r}\tau')$. The target range, $L$, can be recovered from the phase of measured voltage, $- 
\omega_{\rm s}\tau + 
  \omega_{\rm r}\tau' = -\frac{2\omega_{\rm s}L}{c} +   \omega_{\rm r}\tau'$, by compensating for the phase introduced by the known reference signal, $  \omega_{\rm r}\tau'$. 
However, due to phase ambiguity, a distinguishable 
range, $L$, must satisfy
$\omega_{\rm s}\tau = \frac{2L\omega_{\rm s}}{c}  < {2\pi}$.
For  a  
typical WiFi 
signal with
$\omega_{\rm s} = 
2\pi\times 1.5 
{\rm GHz}$, we have $L < 0.1{\rm m}$. This implies that for a \emph{single-tone} sensing signal, the RARE is suitable for estimating target ranges on a small scale, such as micro-vibrations~\cite{QuanSense_Zhang2023}. 
\item \emph{Example 2} ($\vartheta_{\rm s}(t) = \frac{1}{2}\alpha t^2 + \omega_0 t$): 
This paper focuses on sensing target ranges on a larger scale, e.g., $L \in [1{\rm m}, 1000{\rm 
m}]$. By using the LFM waveform, $\vartheta_{\rm s}(t) = \vartheta_0(t) =  \frac{1}{2}\alpha t^2 + \omega_0 t$, as commonly used in classic wireless sensing, the phase ambiguity issue can be effectively addressed~\cite{Radar,malanowski_analysis_2014}. With the LFM waveform, the received waveform is expressed as $y(t)\propto \cos(\frac{1}{2} \alpha(t - \tau)^2 + \omega_0(t - \tau) - \omega_{\rm r}(t - \tau'))$. This enables the recovery of the range $L$ from $y(t)$ using time-frequency analysis techniques.

\end{itemize}


	
	
		

\subsection{Implementation Challenges of Heterodyne Sensing}

Although the heterodyne sensing technique exhibits high sensitivity, it 
 faces two critical challenges in implementation. 
\subsubsection{Bulky transceiver architecture}
Heterodyne sensing requires a bulky transceiver architecture.
As 
presented in Fig.~\ref{img:HeterodyneSensing}(c), 
an external reference source must be deployed, leading to increased hardware cost and energy consumption.  Furthermore, ensuring proper signal decoupling between the transmitter and RARE is essential to mitigate self-interference, which entails additional hardware and algorithm costs.

\subsubsection{Insufficient instantaneous bandwidth}
To satisfy the minimum resolution requirement of range sensing, the sweep 
bandwidth, $B$, of the LFM signal  
in \emph{Example 2} needs to be at least tens of 
MHz~\cite{ksiezyk_opportunities_2023}. However, while the detectable frequency 
range of a RARE is wide due to the abundant energy levels, the detectable 
instantaneous bandwidth of 
RARE 
typically does not exceed 10~MHz~\cite{yang_highly_2024}.
This limitation results from the long response time needed for the quantum 
state $\rho_{12}$ to become steady. The mismatch between the required sweep bandwidth, $B$, and the limited instantaneous bandwidth of RARE hinders the application of existing heterodyne-sensing techniques in high-resolution range 
sensing.



\section{Self-Heterodyne Sensing: Architecture and Properties} \label{sec:3}
\begin{figure}
    \centering
    \includegraphics[width=3.5in]{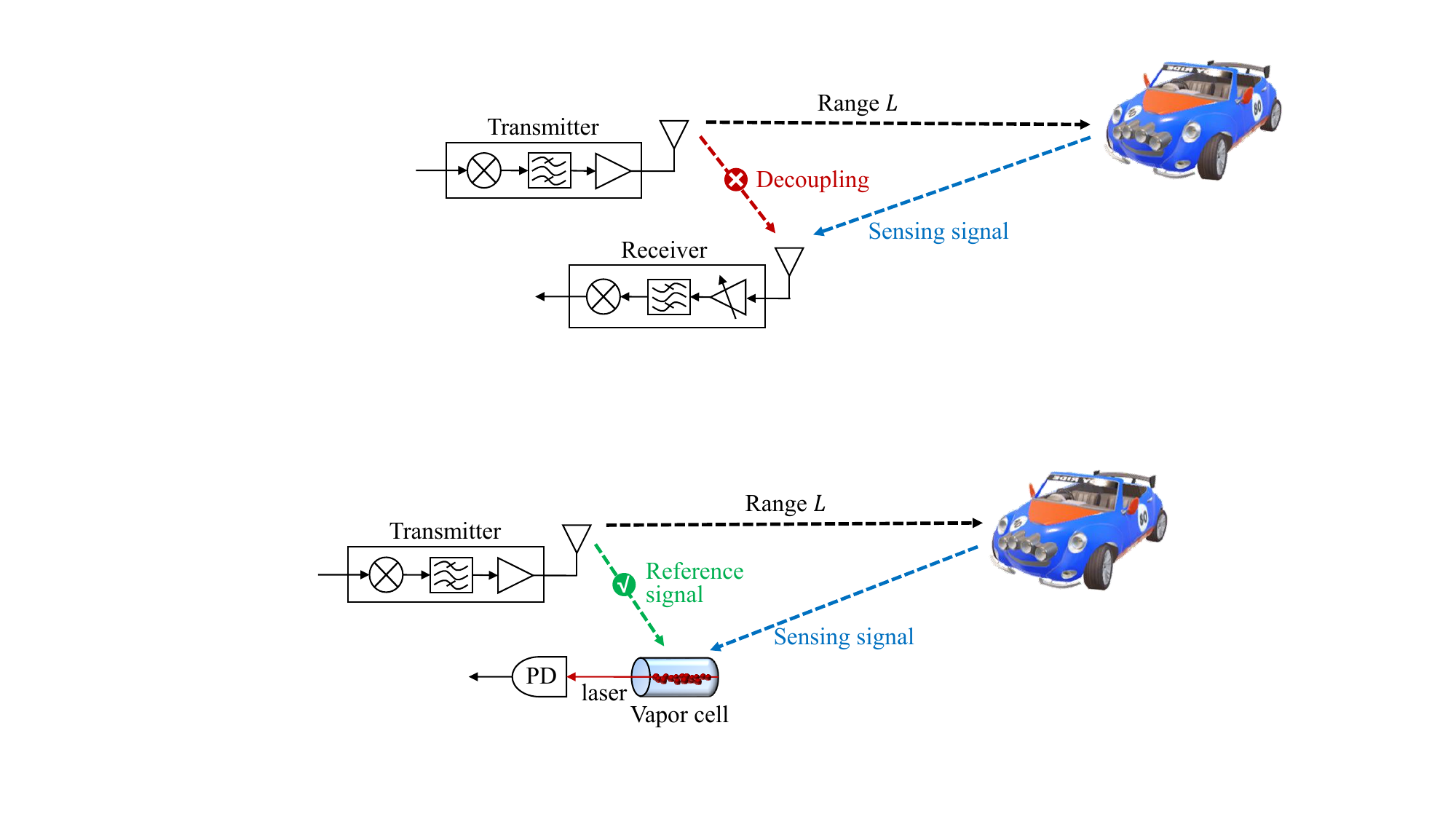}
    \vspace*{-0.5em}
    \caption{The proposed quantum self-heterodyne  sensing system.} 
	\vspace*{-1em}
	\label{img:SelfHeterodyneSensing}
\end{figure}

To address the implementation challenges
of heterodyne sensing, 
we propose the concept of \emph{self-heterodyne sensing} for RAREs.
 Its key idea is to utilize the self-interference generated
by the transmitted sensing signal as the reference
signal. This approach is motivated by two factors: 1) In quasi-monostatic sensing where the transmitter is in close proximity to the RARE, the self-interference signal can be as
strong as the reference signal. 2) The waveform of the transmitted signal is already known to the RARE in sensing applications. With these considerations, the transmitted signal itself is a perfect alternative to the reference, enabling the elimination of the transmitter-receiver decoupling and the external reference source.

In this section, we first introduce the architecture of self-heterodyne sensing.  Subsequently, a comprehensive discussion on the properties of self-heterodyne is conducted to show its superiority over existing heterodyne sensing techniques. 


\subsection{Architecture of Self-Heterodyne Sensing} \label{sec:3.1}
The architecture of self-heterodyne sensing is depicted in 
Fig.~\ref{img:SelfHeterodyneSensing}, which comprises a classical transmitter, 
a RARE, and the target to be sensed. The transmitter broadcasts a signal 
towards the RARE and the target simultaneously. The signal component that 
directly reaches the RARE acts as the reference signal for assisting signal 
detection. The signal component directed towards the target is reflected back 
and serves as the sensing signal. Details regarding the architecture are provided 
below.

\subsubsection{Transmitted signal} 
To enable high-resolution and high-sensitivity range sensing, we use a variant of LFM waveform as the transmitted signal
\begin{align}\label{eq:LFM2}
    s(t) = \sqrt{P_{\rm tx}(t)}e^{j\vartheta_{\rm s}(t)},\:\:0\le t\le T,
\end{align}
where $\vartheta_{\rm s}(t) = \vartheta_0(t) = \frac{1}{2}\alpha t^2 + \omega_0 t$. Different from the LFM waveform adopted in \eqref{eq:LFM} with a \emph{fixed transmission power} $P_{\rm tx}$, we allow the transmission power, $P_{\rm tx}(t)$, to be \emph{time-varying} in \eqref{eq:LFM2}. This design aims at enhancing the sensitivity of RARE, which will be elaborated in Section~\ref{sec:5}. {\color{black} We set the frequency of the time-varying power $P_{\rm tx}(t)$ as the inverse of symbol duration $1/T$ (on the order of $\sim 1\:{\rm kHz}$).
It thus varies slowly relative to the phase function $\vartheta_{\rm s}(t)$, whose instantaneous frequency is on the order of $\sim 3\:{\rm GHz}$.}
Therefore, the instantaneous frequency of $s(t)$ remains $\frac{{\rm d}\vartheta_{\rm s}(t)}{{\rm d}t} = \alpha t + \omega_0$. 

\subsubsection{Received sensing field}  
The sensing signal received at the RARE originates from the transmitted signal reflected by the target.   Its electric field component due to the roundtrip delay $\tau$ is expressed 
\begin{align}
    E_{\rm s}(t) &= |E_{s}(t)|e^{j\vartheta_{\rm s}(t - \tau)} .
\end{align}
Here, $|E_{s}(t)| > 0$ characterizes the time-varying field strength resulting from the time-varying transmission power.  Since Rydberg atoms interact directly with the electric field, accurately modeling $|E_{s}(t)|$ is essential. At time $t$, the incident sensing field delivers a Poynting vector amplitude of 
\begin{align}
    S_{\rm s}(t) = \frac{P_{\rm tx}(t - \tau)G_{\rm tx}A_{\rm c}}{16\pi^2L^4} = P_{\rm tx}(t - \tau)G_{\rm tx} h^2,
\end{align}
where $G_{\rm tx}$ 
represents the transmit antenna gain for the transmitter-to-target link. 
{\color{black}The resulting field strength $|E_{\rm s}(t)|$ is therefore 
\begin{align}
    |E_{\rm s}(t)| &= \sqrt{2Z_0S_{\rm s}(t)} =  \sqrt{2Z_0P_{\rm tx}(t - \tau)G_{\rm tx}h^2} \notag \\ &\overset{(a)}{\approx}  \sqrt{2Z_0P_{\rm tx}(t)G_{\rm tx}h^2},
\end{align}
where $Z_0$ is the vacuum impedance.  Approximation (a) is justified by two temporal scale separations. First, the symbol duration $T$ (on the order of $1\:{\rm ms}$) significantly exceeds the delay $\tau$ (on the order of $1\:{\rm us}$).
In addition, the time-varying amplitude $P_{\rm tx}(t)$ fluctuates at approximately 1~kHz, which is substantially slower than the phase modulation $\vartheta_{\rm s}(t)$ operating at approximately 1~GHz. These disparate time scales permit the safe neglect of the time delay $\tau$ in the power profile. }
This completes the model of the incident sensing field. 

Consequently, the time-varying field strength produces a correspondingly time-varying Rabi frequency for electron transition: 
\begin{align}
    \Omega_{\rm s}(t) = \frac{\mu_{34}}{\hbar} 
	|E_{\rm s}(t)| = \frac{\mu_{34}}{\hbar} 
	\sqrt{2Z_0P_{\rm tx}(t)G_{\rm tx}h^2}.
\end{align}

\subsubsection{Received reference field}
Likewise, the reference field incident on the RARE is a delayed version of the transmitted signal $s(t)$: 
\begin{align}\label{eq:Er2}
    E_{\rm r}(t) &=  |E_{\rm r}(t)|e^{j\vartheta_{\rm s}(t - \tau')},
\end{align}
where $\tau' = \frac{L'}{c}$ refers to the propagation delay from transmitter to receiver, and $L'$ the link distance. 
The reference field strength can be modeled analogously to the sensing field:
\begin{align}
    |E_{\rm r}(t)| &= \sqrt{2Z_0P_{\rm tx}(t - \tau')G_{\rm tx}'h'^2} {\approx}  \sqrt{2Z_0P_{\rm tx}(t)G_{\rm tx}'h'^2},
\end{align}
Unlike the sensing counterpart, the reference field is characterized by a different transmit antenna gain $G'_{\rm tx}$ and channel fading $h'$. The difference in gain arises because the receiver and the sensing target may lie in different directions relative to the transmitter. The channel fading $h'$ is modeled as free-space path loss:  $h' = \frac{1}{\sqrt{4\pi L'^2}}$.

As a result, the Rabi frequency associated with the reference field, $E_{\rm r}(t)$, is given by
\begin{align}
    \Omega_{\rm r}(t) =  \frac{\mu_{34}}{\hbar} 
	\left|E_{r}(t)\right|= \frac{\mu_{34}}{\hbar}\sqrt{2Z_0P_{\rm tx}(t)G_{\rm tx}'h'^2}.
\end{align}

\subsubsection{Overall incident RF field}
The overall incident RF field is the superposition of the sensing and reference field components:
	\begin{align}\label{eq:shs}
		E_{\rm RF}(t)= |E_{\rm r}(t)|e^{j\vartheta_{\rm s}(t - \tau')} + |E_{\rm s}(t)|e^{j\vartheta_{\rm s}(t - \tau)}    .
	\end{align}
Leveraging the expressions of $\Omega_{\rm s}(t)$, $\Omega_{\rm r}(t)$, and $E_{\rm RF}(t)$, we can derive the Rabi frequency associated with the incident RF field as
\begin{align}\label{eq:Rabi}
	\Omega_{\rm RF}(t) &= \frac{\mu_{34}}{\hbar} \left| 
|E_{\rm r}(t)|e^{j\vartheta_{\rm s}(t - \tau')} + |E_{\rm s}(t)|e^{j\vartheta_{\rm s}(t - \tau)}  \right| \notag\\
		&= |\Omega_{\rm r}(t) + 
	\Omega_{\rm s}(t)e^{j\vartheta_{\rm ss}(t; \tau, \tau')}|.
\end{align}
The phase function, 
\begin{align}
	\vartheta_{\rm ss}(t; \tau, \tau') =\vartheta_{\rm s}(t - \tau) - 
	\vartheta_{\rm s}(t - \tau')= - \omega t - \varphi,
\end{align}
characterizes the difference in phase profile between the reference and sensing fields, 
where $\omega \overset{\Delta}{=} (\tau - \tau')\alpha$ and $\varphi 
\overset{\Delta}{=}\frac{1}{2}(\tau 
- \tau')(\omega_0 - \tau - 
\tau')\alpha$. 
Regarding the overall frequency detuning, $\Delta_{\rm RF}$, it is still determined 
by the strong reference signal, $E_{\rm r}(t)$, similar to the conventional heterodyne sensing scheme. 
A key distinction, however, is that the carrier frequency of the reference signal in
\eqref{eq:Er2} is no longer fixed but is instead time-varying. 
This characteristic makes the frequency detuning of self-heterodyne sensing time-varying as well. Specifically,  the instantaneous 
frequency of the reference signal in \eqref{eq:Er2} is given by $ 
\frac{{\rm d}\vartheta_{\rm s}(t - \tau')}{{\rm d}t}$. The 
frequency detuning must therefore be adjusted from \eqref{eq:Delta_RF} to
\begin{align}\label{eq:Delta}
    \Delta_{\rm RF}(t)  &= \Delta_r(t) = \frac{{\rm d} \vartheta_{\rm s}(t - 
    \tau')}{{\rm d} 
    t} - \omega_{34} = \alpha t + \gamma,
\end{align}
where $\gamma \overset{\Delta}{=} - \alpha \tau' + \omega_0 - 
\omega_{34}$. 

\subsubsection{Output voltage and noise}
Substituting \eqref{eq:Rabi} and \eqref{eq:Delta} into $V_{\rm 
out} = V_{\rm in}e^{C_0\Im{\rho_{21}(t)}} = \Pi(\Omega_{\rm RF}, \Delta_{\rm RF})$ yields the measured voltage:  
\begin{align}\label{eq:shsy}
    y(t) &=
    \Pi(|\Omega_{\rm r}(t) + \Omega_{\rm s}(t)e^{j\vartheta_{\rm ss}(t;\tau, \tau')}|, \Delta_r(t)) + 
    n(t) \notag  \\
    &\overset{(a)}{\approx} \underbrace{\Pi(\Omega_{\rm r}(t),  \Delta_r(t))}_{\rm 
    time-varying\:bias} + 
    \underbrace{\frac{\mu_{34}}{\hbar}\Upsilon(\Omega_{\rm r}(t),  
    \Delta_r(t))}_{\rm time-varying\:gain} \notag\\
    &\quad\quad\quad\times |E_s(t)|\cos\left(\vartheta_{\rm ss}(t; \tau, \tau')\right) + 
    n(t),  
\end{align}
where (a) follows from the first-order Taylor expansion 
of $\Pi(\Omega_{\rm RF}(t), \Delta_r(t))$ under the strong reference condition $\Omega_{\rm r}(t) \gg \Omega_{\rm s}(t)$. A comparison of \eqref{eq:shsy} and \eqref{eq:signal_model} reveals that the bias, $\Pi(\Omega_{\mathrm{r}}, \Delta_r)$, and intrinsic gain, $\frac{\mu_{34}}{\hbar}\Upsilon(\Omega_{\mathrm{r}}, \Delta_r)$, are time-invariant in heterodyne sensing due to the fixed frequency detuning of the reference field. In self-heterodyne sensing, however, the time-varying carrier frequency of the reference field causes the corresponding bias and intrinsic gain, denoted by $\Pi(\Omega_{\mathrm{r}}(t), \Delta_r(t))$ and $\frac{\mu_{34}}{\hbar}\Upsilon(\Omega_{\mathrm{r}}(t), \Delta_r(t))$, to become time-varying. This necessitates more advanced signal processing techniques for analyzing the received signal.

{\color{black} For the Gaussian noise $n(t)$, which satisfies $\mathrm{E}(n(t')n(t))=\sigma^2(t)\delta(t 
- t')$, the power spectrum density, $\sigma^2(t)$, is obtained by replacing the bias and intrinsic gain in \eqref{eq:noise_hs} with their self-heterodyne counterparts, $\Pi(\Omega_{\mathrm{r}}(t), \Delta_r(t))$ and $\frac{\mu_{34}}{\hbar}\Upsilon(\Omega_{\mathrm{r}}(t), \Delta_r(t))$. This substitution yields the time-varying power densities of the ETN, ITN, and total noise: 
\begin{align}\label{eq:noise_shs}
    \sigma^2(t)= \underbrace{\frac{\mu_{34}^2}{\hbar^2}\Upsilon^2(\Omega_r(t),\Delta_r(t)) \langle E_I^2 \rangle}_{\sigma_{\rm etn}^2(t)}
+\underbrace{qR_{\rm T}\Pi(\Omega_r(t),\Delta_r(t))}_{{\sigma_{\rm itn}^2(t)}}.
\end{align}
}

\begin{remark} {\rm \textbf{(Wireless sensing principle of self-heterodyne detection)}}
    The signal model
in \eqref{eq:shsy} reveals that the proposed self-heterodyne sensing strategy 
naturally encodes the target range $L = \frac{c\tau}{2}$, 
into the phase difference function $\vartheta_{\rm ss}(t; \tau, \tau') = -\omega t - \varphi$. This conversion allows the detection of target range by estimating the fluctuation frequency of the output voltage, $\omega = (\tau - \tau')\alpha$.
\end{remark}

\begin{figure}
	\centering
	\includegraphics[width=3.5in]{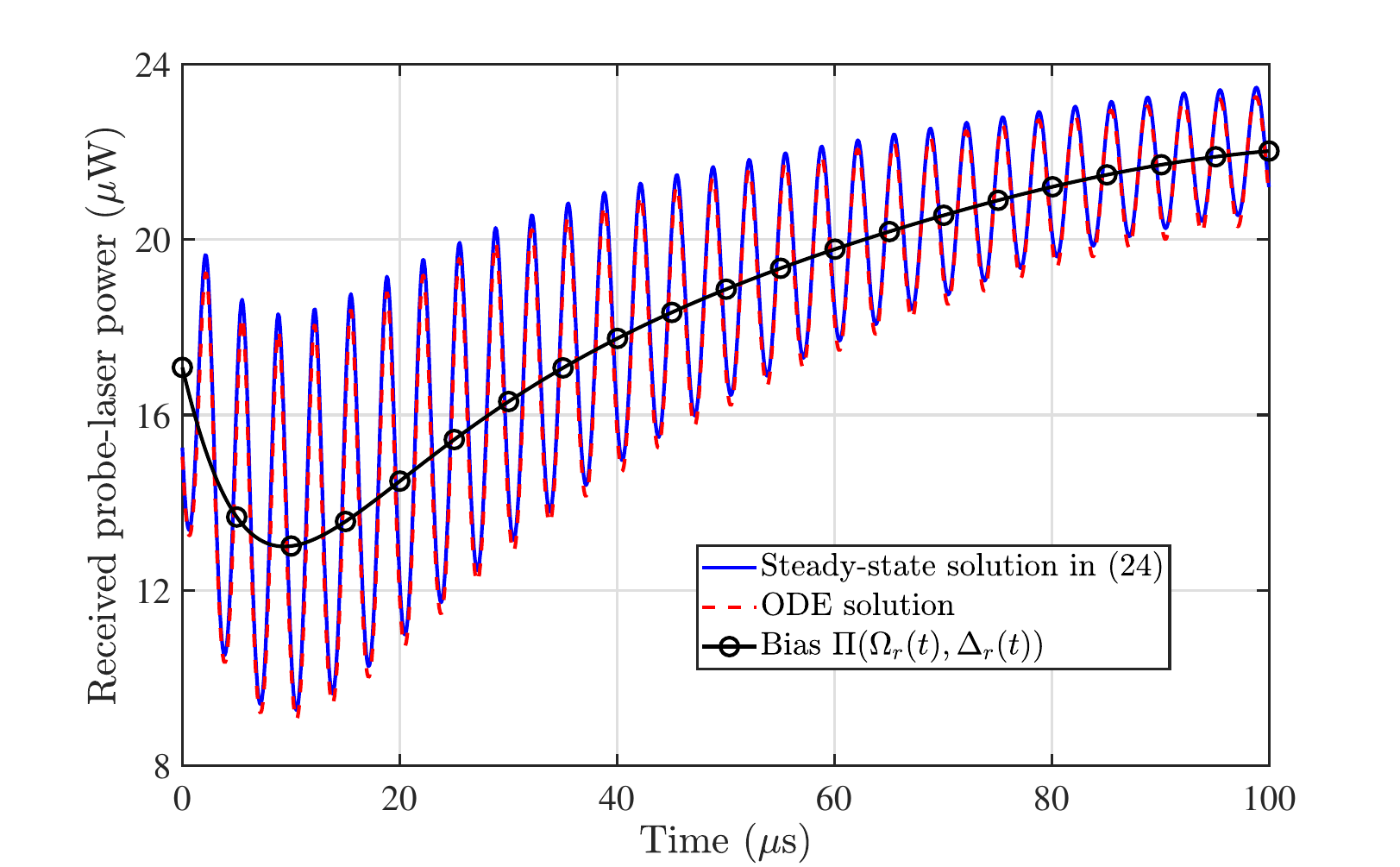}
	\vspace*{-0.5em}
	\caption{The evolution of the measured probe-laser power $y(t)$ in the absence of noise. Some key parameters are listed below: $\Omega_{\rm p} = 2\pi\times6\:{\rm MHz}$, $\Omega_{\rm c} = 2\pi\times 10\:{\rm MHz}$, $\Omega_{\rm r}(t) = 2\pi\times 2 \sqrt{1 + 5t/T}\:{\rm MHz}$, $\Omega_{\rm s}(t) = 2\pi\times 0.1\sqrt{1 + 5t/T} \:{\rm MHz}$, $B = 40\:{\rm MHz}$, $T = 100\:\mu{\rm s}$, and $\tau - \tau' = 0.75\:{\rm \mu s}$.} 
	\vspace*{-1em}
	\label{img:ODEwaveform}
\end{figure}

\subsubsection{Numerical validation}
Figure~\ref{img:ODEwaveform} presents an example of the measured voltage, $y(t)$, in the absence of noise. The curve labeled "ODE solution" is computed using the RydIQule software~\cite{miller_rydiqule_2024}, an open-source Rydberg sensor simulator that employs ordinary differential equation (ODE) solvers to numerically solve the Lindblad master equation \eqref{eq:lindblad}, thereby providing accurate dynamics for $\rho_{21}(t)$ and $y(t)$. The derived steady-state solution for the output voltage from \eqref{eq:shsy} shows close agreement with this ODE solution. Both results exhibit fluctuations around the time-varying bias $\Pi(\Omega_{\mathrm{r}}(t), \Delta_r(t))$ at an identical frequency of $(\tau - \tau')\alpha$. This strong correlation validates the efficacy and accuracy of the derived signal model for self-heterodyne sensing.

\subsection{Properties of Self-heterodyne Sensing}\label{sec:3.3}
We now elaborate on the unique properties and advantages of self-heterodyne sensing. 

\subsubsection{Transceiver architecture} 
The transceiver architecture for self-heterodyne sensing is significantly simpler than that of its heterodyne counterpart. This approach eliminates the need for the deployment of an additional reference source and the signal decoupling between the transmitter and RARE, both of which are essential in heterodyne sensing.

\subsubsection{From atomic mixer to atomic autocorrelator}
The signal model for heterodyne sensing in \eqref{eq:signal_model} shows that the received signal's phase is determined by the sensing waveform down-converted by the reference signal's carrier frequency, i.e., $\vartheta_{\mathrm{sr}}(t; \tau, \tau') = \vartheta_{\rm s}(t - \tau) - \omega_{\mathrm{r}}t + \omega_{\mathrm{r}}\tau'$.
In contrast, for self-heterodyne sensing, the phase function of received signal is the difference between the phase profiles of the transmitted signal at two distinct delays, i.e., $\vartheta_{\rm ss}(t;\tau,\tau') = \vartheta_{\rm s}(t - \tau) - \vartheta_{\rm s}(t - \tau')$. This formulation precisely captures the autocorrelation of the transmitted signal.
This comparison underscores the fundamental physical distinction between the two methods: heterodyne RARE operates as an \emph{atomic mixer}~\cite{RydPhase_Simons2019}, whereas self-heterodyne RARE functions as an \emph{atomic autocorrelator}.

\subsubsection{Bandwidth}
\begin{figure}
	\centering
	\subfigure[Heterodyne-sensing with $\vartheta_{\rm s}(t) = 
	\frac{1}{2}\alpha t^2 + \omega_0 t$.]
	{\includegraphics[width=3.5in]{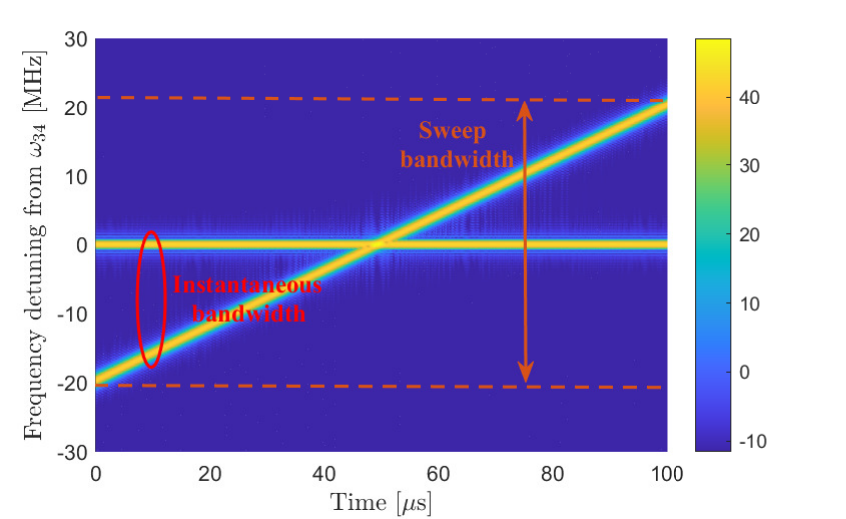}}
	\subfigure[Self-heterodyne-sensing with $\vartheta_{\rm s}(t) = \frac{1}{2}\alpha 
	t^2 + 
	\omega_0t$.]{\includegraphics[width=3.5in]{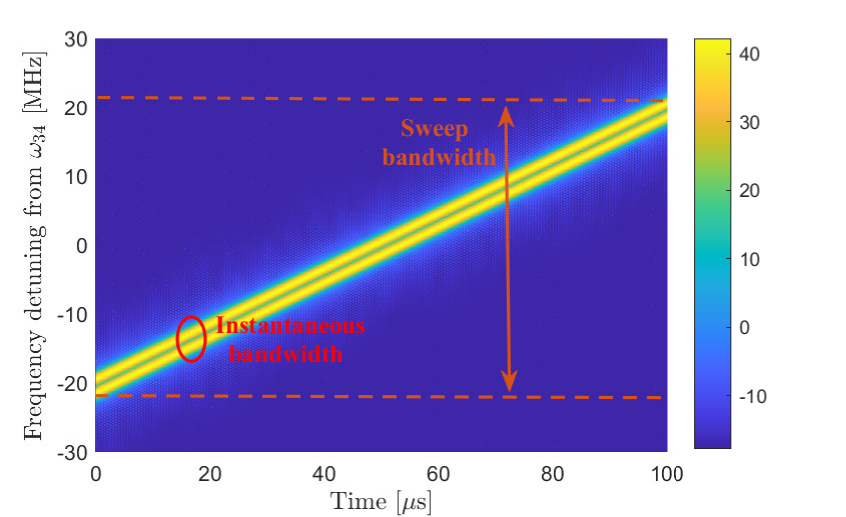}} 
	\caption{STFT of incident RF signals. The parameters are as follows: $T 
		= 100{\mu s}$, $\tau' = 10{\rm n}s$, $\tau = 5{\rm \mu s}$, $B = 
		40 
		{\rm 
			MHz}$, and  
		$\omega_0 = \Omega_{\rm r} = 
		3.5{\rm 
			GHz}$. } %
	\vspace*{-1em}
	\label{img:STFT}
\end{figure}

Self-heterodyne sensing can support a much wider sweep bandwidth, $B$, compared 
to heterodyne sensing. This is because the atomic autocorrelator enables the 
control of the instantaneous bandwidth of the received signal to fit the 
detectable bandwidth range of a RARE, as elaborated below. 

In the context of a classical atomic mixer realized by heterodyne sensing, the
incident RF signal described in \eqref{eq:hs} consists of a single-frequency 
reference signal and a LFM sensing 
signal. Fig.~\ref{img:STFT}(a) visualizes  the 
Short-Time-Fourier-Transform (STFT) 
of this RF signal~\cite{boashash_estimating1_1992}. The frequency gap between the LFM component and the single-frequency component varies linearly over  time.  Particularly, the instantaneous bandwidth of the incident RF signal can expand from 0 to half of the sweep bandwidth, $B/2$, as time going by.  
In a typical wireless sensing application, $B$ can range from tens of MHz to tens of GHz. Unfortunately, it is challenging for a RARE to detect such a broad instantaneous bandwidth.

In the proposed atomic autocorrelator implemented by self-heterodyne sensing, the incident RF signal in  \eqref{eq:shs} is 
the superposition of two identical waveforms with different delays, $\tau$ and 
$\tau'$. This configuration ensures that the frequency gap between the two signal 
components remains unchanged over time, resulting in a constant and narrow instantaneous bandwidth. 
The STFT 
of RF signal is plotted in 
Fig.~\ref{img:STFT}(b). 
Despite having a broad sweep bandwidth, the 
 RF signal only occupies a narrow frequency range instantaneously due to the identical reference and sensing waveforms. 
 Mathematically, the instantaneous bandwidth of the received signal in \eqref{eq:shsy} is given as $\frac{\omega}{2\pi} = \frac{(\tau - \tau')\alpha}{2\pi} = \frac{\tau - \tau'}{T}B$. As the symbol duration $T$ is a controllable parameter, we can properly set its value to make $\frac{\tau - \tau'}{T}B$  fall within the detectable bandwidth range of a RARE. 
 For instance, consider typical values in wireless sensing, $\tau = 1\:{\rm \mu s}$, $\tau' = 1\:{\rm ns}$, and $T = 100\:{\rm \mu s}$. The instantaneous bandwidth occupies just $1\%$ of the sweep bandwidth, $B$. 
Thereafter, the proposed self-heterodyne sensing is more suitable for RARE in high-resolution wireless sensing.   

\section{Self-Heterodyne Sensing: Algorithmic Design}\label{sec:4}
In this section, we first formulate the range detection problem in  
self-heterodyne sensing as an NLS optimization. Then, a 
two-stage algorithm is designed to 
solve the formulated problem. This algorithm involves a coarse estimation stage based on spectrum analysis followed by a refinement stage using Newton's method. 
Lastly, the CRLB for self-heterodyne sensing will be analyzed.

\subsection{NLS Problem Formulation}

To ease the discussion, we normalize the received signal as 
\begin{align}\label{eq:ybar}
	\bar{y}(t) &= \frac{y(t) - \Pi(\Omega_{\rm r}(t),  \Delta_r(t))}{\sigma(t)} \notag\\&= 
	\varrho(t) h\cos(\omega t + \varphi) + 
	\bar{n}(t),
\end{align}
where the equivalent Gaussian noise $\bar{n}(t)$ has unit power, {and} the time-varying amplitude $\varrho(t)$ is defined as
\begin{align}\label{eq:varrho}
	\varrho(t) = \frac{ \mu_{34}\Upsilon(\Omega_{\rm r}(t),  
    \Delta_r(t)) \sqrt{2Z_0P_{\rm tx}(t)G_{\rm tx}}}{\hbar \sigma(t)}
\end{align}
As the relative position from the transmitter to the RARE is fixed, the parameters of the reference signal, $\tau'$, $\Omega_{\rm r}(t)$, $\Delta_r(t)$, and the time-varying amplitude, 
$\varrho(t)$, can be acquired in 
advance. Conversely, the channel parameters pertaining to the target, i.e., $h$, 
$\tau$, and $\varphi$, are unknown to RARE and should be estimated.
As the target range, $L$, is embedded in the frequency of the received signal,  
$\omega = (\tau - \tau')\alpha = (\frac{2L}{c} - \tau')\alpha$, the range 
sensing 
problem is equivalently a frequency estimation problem with a time-varying 
amplitude, $\varrho(t)$. We adopt the NLS principle to formulate 
it as
\begin{align}\label{eq:P0}
	\min_{\hat{h}, \hat{\omega}, \hat{\varphi}} \int_{0}^T |\bar{y}(t) - 
	\varrho(t) \hat{h}\cos(\hat{\omega} t + \hat{\varphi})|^2{\rm d}t,
\end{align} 
Solving problem \eqref{eq:P0}, we can estimate the propagation delay and target range as 
$\hat{\tau} = \frac{\hat{\omega}}{{\alpha}} + \tau'$ and 
$\hat{L} = \frac{1}{2}c\hat{\tau}$, respectively. 

\subsection{NLS Algorithm}
We first seek the optimal solution for $\hat{h}$ in the formulated NLS problem. 
By setting the derivative of 
\eqref{eq:P0} w.r.t $\hat{h}$ to zero, the optimal $\hat{h}$ can be 
obtained 
straightforwardly as
\begin{align}\label{eq:h_opt}
\hat{h} = \frac{\int_{0}^T \bar{y}(t)  
	\varrho(t)\cos(\hat{\omega} t + \hat{\varphi}){\rm 
		d}t}{\int_{0}^T |  
	\varrho(t)\cos(\hat{\omega} t + \hat{\varphi})|^2{\rm 
		d}t}.
\end{align}
By further substituting \eqref{eq:h_opt} into \eqref{eq:P0}, the original NLS problem is transformed to
\begin{align}\label{eq:P1}
	\max_{\hat{\omega}, \hat{\varphi}}\: Q(\hat{\omega}, \hat{\varphi}) 
	\overset{\Delta}{=} 
	\frac{\left|\int_{0}^T \bar{y}(t)  
		\varrho(t)\cos(\hat{\omega} t + \hat{\varphi}){\rm 
			d}t\right|^2}{\int_{0}^T |  
			\varrho(t)\cos(\hat{\omega} t + \hat{\varphi})|^2{\rm 
				d}t},
\end{align}
which is clearly a non-convex programming. Finding its global optimal 
solution requires a high-complexity two-dimensional search over the frequency 
and phase plane.
%
We devise a two-stage algorithm to address \eqref{eq:P1} with low complexity, which involves a coarse estimation stage followed by a refinement
stage. The first stage aims at finding good initial estimates,  
$\hat{\omega}^{(0)}$ and $\hat{\phi}^{(0)}$, efficiently, while the refinement stage fine-tunes $\hat{\omega}$ and $\hat{\phi}$  to approach the 
global optimal solution. The step-by-step procedures are presented below.

\subsubsection{Coarse estimation stage}
{\color{black}The first stage relies on the assumption that cosine term $\cos(\hat{\omega}t + \hat{\phi})$ in $Q(\hat{\omega}, \hat{\varphi})$ oscillates much faster than the amplitude term $ 
\varrho(t)$.  This timescale separation is justified by comparing their characteristic frequencies. The amplitude $ 
\varrho(t)$ fluctuates at a rate on the order of the inverse of symbol duration $\frac{1}{T}$.  The cosine term oscillates at a frequency of $\frac{\omega}{2\pi} = \frac{\tau B}{T}$.
Using typical system parameters: ($T\sim 1\:{\rm ms}$, $\tau \sim 1\:{\mu s}$, and $B\sim{100\:{\rm MHz}}$), we find:
    $\frac{1}{T}\sim 1\:{\rm kHz} \ll \frac{\tau B}{T} \sim 100\:{\rm kHz}$.
The significant disparity (a factor of 100 in this example) validates the timescale separation.}
Leveraging this frequency separation, we approximate the denominator in 
\eqref{eq:P1} as 
\begin{align}\label{eq:approx}
 &\int_{0}^T |  
	\varrho(t)\cos(\hat{\omega} t + \hat{\varphi})|^2{\rm 
		d}t= \frac{1}{2} \int_{0}^{T} | 
\varrho(t)|^2{\rm 
d}t \notag\\&- \frac{1}{2} \int_{0}^{T} | 
\varrho(t)|^2\cos(2\hat{\omega}t + 2\hat{\phi}){\rm 
d}t {\approx}   \frac{1}{2} \int_{0}^{T} | 
\varrho(t)|^2{\rm 
d}t.
\end{align}
The approximation follows because the second, rapidly oscillating integral averages to zero over the symbol duration $T$.
Using \eqref{eq:approx}, the 
problem \eqref{eq:P1} 
can be simplified as 
\begin{align}\label{eq:P2}
	\max_{\hat{\omega}^{(0)}, \hat{\varphi}^{(0)}} 
	\left|{\int_{0}^T \bar{y}(t)  
		\varrho(t)\cos(\hat{\omega}^{(0)} t + \hat{\varphi}^{(0)}){\rm 
			d}t}\right|^2.
\end{align}
The initial estimated phase, $\hat{\phi}^{(0)}$, can thus be attained by setting 
the
derivative of 
\eqref{eq:P2} w.r.t $\hat{\varphi}^{(0)}$ to zero:
\begin{align}\label{eq:phi}
	\tan\left(\hat{\varphi}^{(0)}\right) = - \frac{{\int_{0}^T \bar{y}(t)  
			\varrho(t)\sin(\hat{\omega}^{(0)} t){\rm 
				d}t}}{{\int_{0}^T \bar{y}(t)  
			\varrho(t)\cos(\hat{\omega}^{(0)} t){\rm 
			d}t}}. 
\end{align}
Then, by substituting \eqref{eq:phi} back to \eqref{eq:P2} and conducting some tedious calculations, we can rewrite the problem \eqref{eq:P2} as 
\begin{align}\label{eq:P3}
	&\max_{\hat{\omega}^{(0)}} 
	\left|{\int_{0}^T \bar{y}(t)  
		\varrho(t)e^{j\hat{\omega}^{(0)} t}{\rm 
			d}t}\right|^2.
\end{align}
One can observe from \eqref{eq:P1} and \eqref{eq:P3} that the original 
two-dimensional optimization problem is transferred to a simpler 
one-dimensional one. 
In particular, problem \eqref{eq:P3} is aimed at searching the principal 
frequency component, $\hat{\omega}^{(0)}$, inside the signal 
$\bar{y}(t)\varrho(t)$. Therefore, it can be efficiently solved using spectrum analysis tools, such as the Fast Fourier Transform (FFT) technique.

\subsubsection{Refinement stage}
The assumption we made in \eqref{eq:approx} might lead to inaccurate estimate of the target range. To address this issue, we introduce the refinement stage. It iteratively fine-tunes the parameters, $\hat{\omega}^{(0)}$ and $\hat{\varphi}^{(0)}$, obtained in the first stage to maximize the original objective function, $Q(\hat{\omega}, \hat{\varphi})$, in \eqref{eq:P1}. 
We use the superscript $(i)$ to label the parameters, $\hat{\omega}^{(i)}$ and 
$\hat{\varphi}^{(i)}$, obtained in the $i$-th iteration. The refinement is conducted using the 
Newton's method, expressed as 
\begin{align}\label{eq:Newton}
	\left(
	\begin{array}{c}
	\hat{\omega}^{(i)} \\
	\hat{\varphi}^{(i)}
	\end{array}
	\right) = 
		\left(
	\begin{array}{c}
		\hat{\omega}^{(i - 1)} \\
		\hat{\varphi}^{(i - 1)}
	\end{array}
	\right) + 			\left(
	\begin{array}{cc}
	 \frac{\partial^2 Q}{\partial \hat{\omega}^2} & 
	 \frac{\partial^2 Q}{\partial \hat{\omega} \partial \hat{\varphi}}\\
	 \frac{\partial^2 Q}{\partial \hat{\varphi} \partial \hat{\omega} } & 
	\frac{\partial^2 Q}{\partial \hat{\varphi}^2}
	\end{array}
	\right)^{-1}	\left(
	\begin{array}{c}
		\frac{\partial Q}{\partial \hat{\omega}} \\
		\frac{\partial Q}{\partial \hat{\varphi}}
	\end{array}
	\right). 
\end{align}
We accept a refinement only if the value of the new objective function 
$Q(\hat{\omega}^{(i)}, 
\hat{\varphi}^{(i)})$ is larger than the old one by a threshold. 

After the Newton refinement converges, the propagation delay and target range are determined as \begin{align}
    \hat{\tau} 
= 
\frac{\hat{\omega}^{(N_{\rm iter})}}{{\alpha}} + \tau'\:{\rm and}\:\hat{L} 
= \frac{1}{2}c\hat{\tau},
\end{align}
respectively, where $N_{\rm iter}$ denotes the total number of iterations.

\subsection{CRLB Analysis}

We now analyze the CRLB of self-heterodyne sensing, as it characterizes the 
lower bound of estimation error. We use the vectors $\boldsymbol{\kappa} = 
[\kappa_1, \kappa_2, \kappa_3]^T \overset{\Delta}{=}[h, 
\omega, 
\varphi]^T$ and $\hat{\boldsymbol{\kappa}} = 
[\hat{\kappa}_1, \hat{\kappa}_2, \hat{\kappa}_3]^T \overset{\Delta}{=}[\hat{h}, 
\hat{\omega}, 
\hat{\varphi}]^T$ to 
represent the true values of unknown parameters and their estimated results. 
In 
particular, $\hat{\tau} =  
\frac{{\hat{\kappa}_2}}{\alpha} + \tau'$ and ${\tau} = 
\frac{{\kappa_2}}{\alpha} + \tau'$. The covariance matrix of 
$\boldsymbol{\kappa}$ is lower bounded by the inverse of Fisher information 
matrix (FIM)
\begin{align}
	{\rm Cov}(\hat{\boldsymbol{\kappa}}) \ge {\rm 
	CRLB}(\boldsymbol{\kappa}) = \mb{I}^{-1}(\boldsymbol{\kappa}),
\end{align}
where $[\mb{I}(\boldsymbol{\kappa})]_{ij} = -E\left(\frac{\partial^2 \log p(y(t)|\boldsymbol{\kappa})}{\partial \kappa_i \partial \kappa_j}\right)$. Here, the loglikehood function is given as 
\begin{align}
	\log p(\bar{y}(t)|\boldsymbol{\kappa}) \propto -\frac{1}{2}\int_{0}^T|\bar{y}(t) - \varrho(t)h\cos(\omega t 
	+ \varphi)|^2{\rm d} t.
\end{align}
The full expression of the FIM  
 is derived in Appendix A for reader's reference. 
Moreover, we are interested in the estimation error of the propagation delay, $\tau$. 
The 
following proposition provides an asymptotic expression of ${\rm  CRLB}(\tau)$ for large $\omega$. 
\begin{proposition}\label{proposition1}
	For large $\omega$, the CRLB of $\tau$ is asymptotically
	\begin{align}\label{eq:CRLB_R}
		{\rm Var}(\hat{\tau}) \ge {\rm CRLB}(\tau) \rightarrow
		\frac{2\bar{\varrho}_0}{\alpha^2h^2(\bar{\varrho}_0\bar{\varrho}_2 - 
				\bar{\varrho}_1^2)},
	\end{align}	
where $\bar{\varrho}_0 \overset{\Delta}{=} 
\int_0^T\varrho^2(t) {\rm d}t$, $\bar{\varrho}_1 \overset{\Delta}{=} 
\int_0^T\varrho^2(t) t {\rm d}t$, and $\bar{\varrho}_2 \overset{\Delta}{=} 
\int_0^T\varrho^2(t) t^2 {\rm d}t$.
\end{proposition}
\begin{proof}
	(See Appendix A). 
\end{proof}
In Section~\ref{sec:6}, numerical results will show that the propagation delay estimated by the proposed algorithm could approach CRLB.

{\color{black}\subsection{Extension to Multi-target Scenario}
The proposed self-heterodyne sensing framework can be naturally generalized to multi-target scenarios. This subsection outlines this extension, demonstrating the method's inherent scalability.
Consider the coexistence of $M$ targets.  We introduce the subscript $m$ to denote parameters for the $m$-th target, including its sensing field, $E_{{\rm s}, m}(t)$, distance $L_m$, propagation delay $\tau_n = \frac{2L_m}{c}$, channel fading $h_m$, and phase shift $\varphi_m$. Note that the channel parameters for the reference field remain unchanged as we make no difference to the reference source. Consequently, the aggregate incident RF field comprises one reference field and $M$ sensing fields:
\begin{align}
    E_{\rm RF}(t) = |E_{\rm r}(t)|e^{j\vartheta(t - \tau')} + \sum_{m = 1}^M|E_{{\rm s},m}(t)|e^{j\vartheta(t - \tau_m)}. 
\end{align}
Following the derivations in Subsection~\ref{sec:3.1}, the resulting output voltage is given by:  
\begin{align}\label{eq:ybar2}
    y(t) &= \Pi(\Omega_{\rm r}(t),  \Delta_{\rm r}(t)) + 
    \frac{\mu_{34}}{\hbar}\Upsilon(\Omega_{\rm r}(t),  
    \Delta_r(t)) \notag\\
    &\quad\quad\quad\times \sum_{m = 1}^M|E_{{\rm s}, m}(t)|\cos\left(\vartheta_{{\rm s}}(t; \tau_m, \tau')\right) + 
    n(t).
\end{align}
The corresponding normalized output voltage then becomes a multi-tone signal:
\begin{align}
    \bar{y}(t) = \rho(t)\sum_{m = 1}^M h_{m}\cos(\omega_mt + \varphi_m) + \bar{n}(t). 
\end{align}
where each frequency component $\omega_m = (\tau_m - \tau')\alpha = \left(\frac{2L_m}{c} - \tau'\right)\alpha$ uniquely encodes the range $L_m$ of the $m$-th target. 
This generalized signal model preserves the structure of the single-target case, with the time-varying amplitude $\rho(t)$ and noise $\bar{n}(t)$ remaining identical to those in the original model.
As a result, the sensing problem naturally transitions from single-frequency estimation to a multi-frequency estimation problem, formulated as: 
\begin{align}\label{eq:P0}
	\min_{\{\hat{h}_m\}, \{\hat{\omega}_m\}, \{\hat{\varphi}_m\}} \int_{0}^T |\bar{y}(t) - 
	\varrho(t) \sum_{m}\hat{h}_m\cos(\hat{\omega}_m t + \hat{\varphi}_m)|^2{\rm d}t.
\end{align} 
This multi-target problem can be efficiently solved by generalizing the proposed NLS algorithm in two stages:
\begin{itemize}
    \item {\bf Coarse estimation}: The top-1 frequency estimation in \eqref{eq:P3} is extended to identifying the top-$M$ frequency peaks of 
    $
    \left|{\int_{0}^T \bar{y}(t)  
		\varrho(t)e^{j\hat{\omega}_m t}{\rm 
			d}t}\right|^2$.
This can be efficiently implemented using algorithms such as FFT or Orthogonal Matching Pursuit (OMP).
\item {\bf Refinement}: The refinement stage is straightforwardly extended by augmenting the variable set in the Newton's solver from the single pair $(\omega, \varphi)$ to the sets $\{\omega_m\}$ and $\{\varphi_m\}$. 
\end{itemize}

In summary, this extension underscores the scalability of the proposed design, demonstrating its capability to handle multi-target sensing scenarios through a natural generalization of both the signal model and the estimation algorithm.
}

\section{Self-Heterodyne Sensing: Sensitivity Maximization}\label{sec:5}
In this section, we propose to maximize the RARE's sensitivity through the manipulation of the time-varying transmission power, $P_{\rm tx}(t)$, which is referred to as \emph{power-trajectory ($P$-trajectory)} design. 
We first justify the necessity of $P$-trajectory design for sensitivity maximization. 
{Subsequently, an ITN-limited $P$-trajectory is developed by deriving the asymptotically optimal transmission power in the presence of ITN only. Finally, this initial design is adapted into a practical $P$-trajectory by incorporating the influence of both ETN and ITN. }

\subsection{Necessity of $P$-trajectory Design}
From the signal model in \eqref{eq:ybar}, the sensing accuracy of RARE is primarily determined by the time-varying amplitude $\varrho(t)$. A larger $\varrho^2(t)$ increases the SNR of RARE, defined as ${\rm SNR}(t) = \varrho^2(t)h^2T$, thereby enhancing sensing sensitivity. This observation motivates us to carefully manipulate the trajectory of the time-varying transmission power, $P_{\rm tx}(t)$, to maximize $\varrho^2(t)$ and achieve the highest possible sensitivity of a RARE. 

Considering noise contributions from both ETN and ITN, $\sigma_{\rm etn}^2(t)$ and $\sigma_{\rm itn}^2(t)$, the SNR can be explicitly expressed as
\begin{align}\label{eq:SNR}
	{\rm SNR}(t) 
= \varrho^2(t)h^2T = \frac{{\rm SNR}_{\rm etn}(t) {\rm SNR}_{\rm itn}(t)}{{\rm SNR}_{\rm etn}(t)+{\rm SNR}_{\rm itn}(t)}.
\end{align}
Here, ${\rm SNR}_{\rm etn}(t)$ represents the ETN-limited SNR: 
\begin{align}\label{eq:SNR_ext}
    {\rm SNR}_{\rm etn}(t) &= \frac{\frac{\mu_{34}^2}{\hbar^2}\Upsilon^2(\Omega_{\rm r}(t),  
		\Delta_r(t))|E_s(t)|^2T}{\sigma_{\rm etn}^2(t)}  \notag \\
        &= \frac{T}{\langle E_I^2\rangle}\underbrace{|E_{\rm s}(t)|^2}_{\rm Field\:strength}, 
\end{align}
while ${\rm SNR}_{\rm itn}(t)$ refers to the ITN-limited SNR: 
\begin{align}\label{eq:SNR_int}
        {\rm SNR}_{\rm itn}(t) &= \frac{\frac{\mu_{34}^2}{\hbar^2}\Upsilon^2(\Omega_{\rm r}(t),  
		\Delta_r(t))|E_s(t)|^2T}{\sigma_{\rm itn}^2(t)}  \notag \\&= \underbrace{\frac{\mu_{34}^2\Upsilon^2(\Omega_{\rm r}(t),  
		\Delta_r(t))T}{q\hbar^2\Pi(\Omega_{\rm r}(t),  
			\Delta_r(t))}}_{\rm Gain\:of\:RARE} \underbrace{|E_s(t)|^2}_{\rm Field\:strength} . 
\end{align}
These two kinds of SNR exhibit distinct dependence on the transmission power, which are compared below.

\begin{remark} {\rm \textbf{(Impact of transmission power on  ETN-limited SNR)}}
In the ETN-limited regime, ${\rm SNR}_{\rm etn}$ scales linearly with the incident field strength, $|E_{\rm 
s}(t)|^2$ (and thus with transmission power $P_{\rm tx}(t)$). This is because the RARE amplifies both the sensing signal and external noise. This linear relationship is consistent with the SNR behavior of a classical receiver. 
\end{remark}
\begin{remark} {\rm \textbf{(Impact of transmission power on  ITN-limited SNR)}}
In the ITN-limited regime, the transmission power, $P_{\rm 
tx}(t)$, and instantaneous frequency, $\frac{{\rm d}\vartheta_{\rm s}(t)}{{\rm d}t}$, of the transmitted signal affect ${\rm SNR}_{\rm itn}$ through two mechanisms:
\begin{itemize}
    \item First, similar to the ETN-limited case, higher transmit power directly increases the strength of the sensing field $|E_s(t)|^2$. 
    \item {\color{black} Second, $P_{\rm tx}(t)$ and $\frac{{\rm d}\vartheta_{\rm s}(t)}{{\rm d}t}$ influence the Rabi frequency, $\Omega_{\rm r}(t) =  \frac{\mu_{34}}{\hbar}|E_{\rm r}(t)|$, and frequency detuning, $\Delta_{\rm r}(t) = \frac{{\rm d}\vartheta_{\rm s}(t - \tau')}{{\rm d}t} - \omega_{34}$, of the reference field, which collectively determine the gain of RARE. Notably, this gain does not increase monotonically with $\Omega_{\rm r}(t)$ or $\Delta_{\rm r}(t)$. Taking the Rabi frequency as an example, it quantifies the intensity of electron transition. 
    When $\Omega_{\rm r}(t)$ is either too low or too high, the electron population becomes predominantly localized in either the initial or final Rydberg states ($\ket{3}$ or $\ket{4}$). In both cases, the system responds poorly to the weak sensing signal. Maximal gain is achieved only at intermediate Rabi frequencies (corresponding to intermediate transmission power), where a balanced electron population exists between Rydberg states. This mechanism causes ${\rm SNR}_{\rm itn}(t)$ to be a non-monotonic function of transmission power. }
\end{itemize}
\end{remark}

\begin{figure}
	\centering
	\subfigure[Received SNR (dB) of a classical receiver.]
	{\includegraphics[width=3.5in]{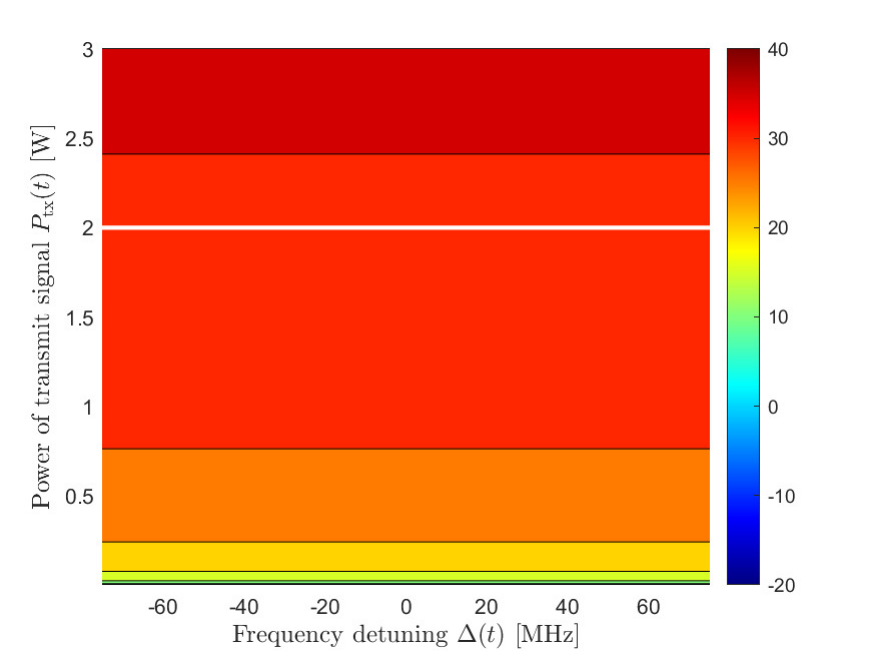}}
	\subfigure[Received SNR (dB) of a 
	RARE.]{\includegraphics[width=3.5in]{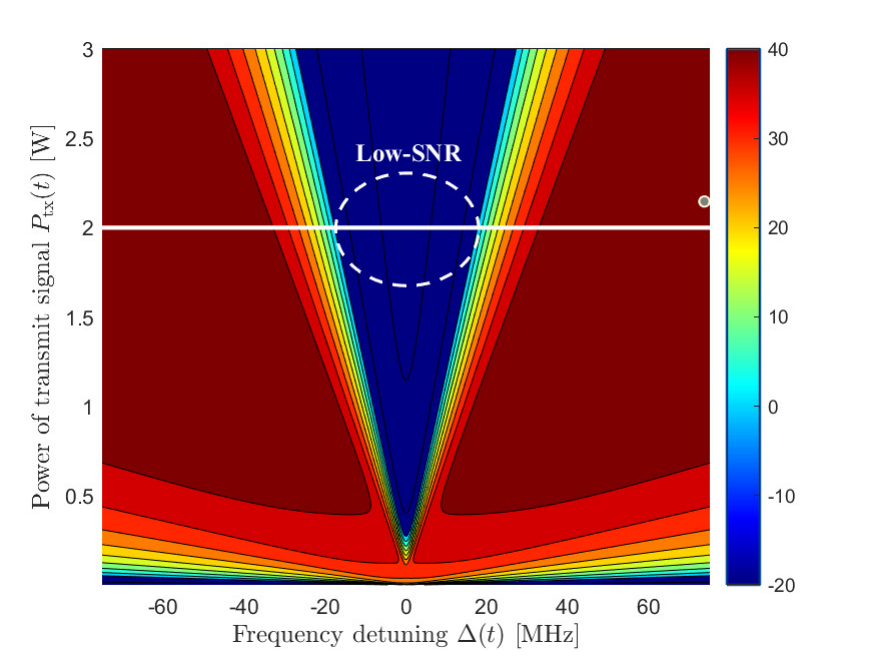}} 
	\caption{The effect of transmission power $P_{\rm tx}(t)$ and frequency 
		detuning $\Delta(t)$ on the received SNR for (a) a classical receiver 
		and 
		(b) a RARE. The white line refers to a fixed transmit power of $P_{\rm tx}(t)\equiv2\:{\rm W}$
	} %
	\vspace*{-1em}
	\label{img:SNR}
\end{figure}

As a result, due to the combined effects of ETN and ITN, the received SNR in self-heterodyne sensing exhibits a complicated pattern with the transmission power, $P_{\rm tx}(t)$, and frequency detuning, $\Delta(t)$.
Fig.~\ref{img:SNR} illustrates this pattern. In the case of a classical receiver, as the frequency detuning of LFM wavefrom sweeps from $-B/2$ to $B/2$, the received SNR remains constant when the transmission power is fixed. Henceforth, a constant transmission power is optimal for classical reception. 
In contrast, a fixed transmission power may cause RARE to experience regions of extremely low SNR during frequency detuning sweeps. For example, with $P_{\rm tx}(t) \equiv 2\:{\rm W}$, RARE incurs over 
50 dB SNR loss as frequency detuning varies from -60 MHz to 0 MHz. 
That is to say, achieving maximal sensitivity in self-heterodyne sensing requires careful optimization of a time-varying $P_{\rm tx}(t)$ trajectory.

		



\subsection{ITN-limited $P$-trajectory Design}

\begin{figure}
	\centering
	\subfigure[RARE's SNR (dB) with ITN only.]
	{\includegraphics[width=3.5in]{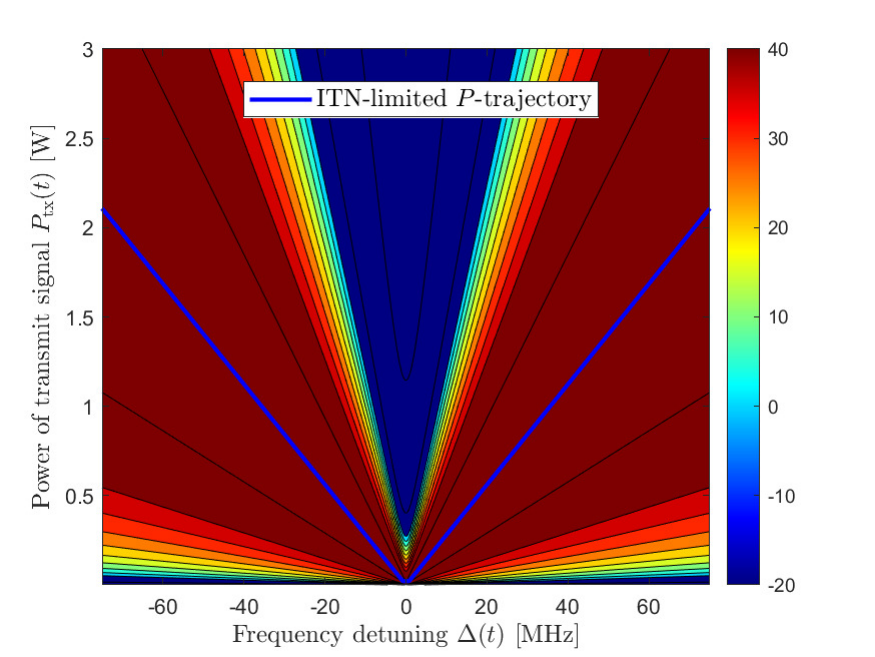}}
	\subfigure[RARE's SNR (dB) with both ETN and ITN.]{\includegraphics[width=3.5in]{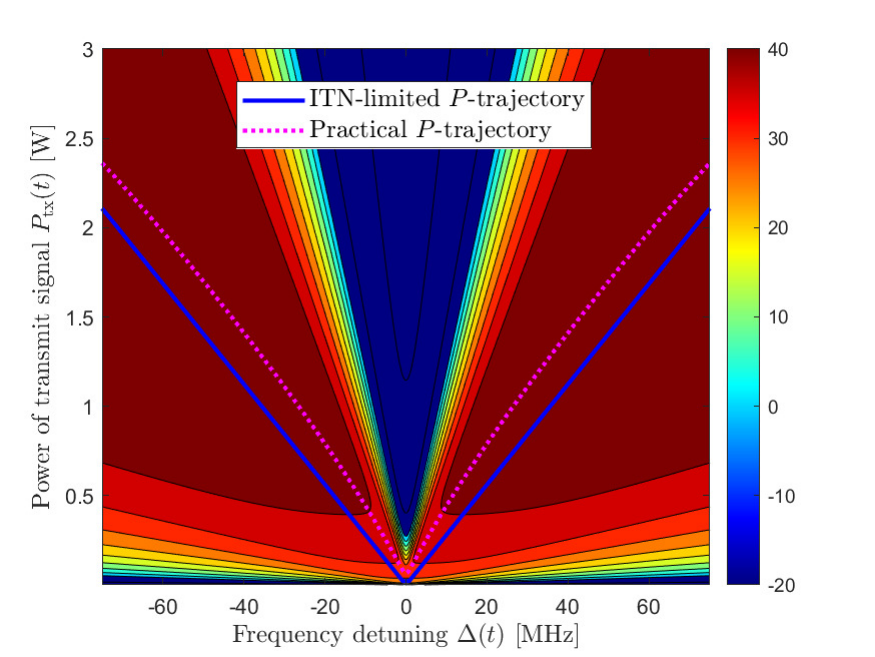}} 
	\caption{(a) The designed ITN-limited $P$-trajectory for RARE with ITN only and (b) the optimized practical $P$-trajectory for RARE with both ITN and ETN. 
	} %
	\vspace*{-1em}
	\label{img:PT}
\end{figure}
We begin by designing a $P$-trajectory for the ITN-limited region, as the RARE's gain in this regime is primarily responsible for the complicated SNR pattern. Since ${\rm SNR}_{\rm itn}$ exhibits an intricate dependence on  $P_{\rm tx}(t)$ and $\Delta(t)$, finding the optimal power $P_{\rm tx}^\star(t)$ that maximizes the ITN-limited SNR for 
arbitrary frequency detuning $\Delta(t)$  is challenging. 
Fortunately, we can determine optimal power values in certain special cases.

As depicted in Fig.~\ref{img:PT}(a), the highest ${\rm SNR}_{\rm itn}$ appears in regions (the dark red area) whereas both the 
frequency detuning and transmission power are large. 
We aim to characterize this high-SNR region at first. 
Specifically, the optimization of transmission power is equivalent to the optimization of Rabi frequency due to the one-to-one mapping 
$\Omega_{\rm r}(t) =  \frac{\mu_{34}}{\hbar}\sqrt{2Z_0P_{\rm tx}(t)G_{\rm tx}'h'^2}$. It becomes convenient to rewrite ${\rm SNR}_{\rm itn}$ as a function of $\Omega_{\rm r}(t)$ and $\Delta_{\rm r}(t)$:
    \begin{align}\small{\rm SNR}_{\rm itn}(t)\propto \left.\kappa^2(\Omega, \Delta)\right|_{\Omega = \Omega_{\rm r}(t),\Delta = \Delta_r(t)},
    \end{align}
    where
    \begin{align}
        \kappa(\Omega, \Delta) 
    \overset{\Delta}{=} \frac{|\Upsilon(\Omega, \Delta)\Omega|}{\sqrt{\Pi(\Omega, \Delta)}}.
    \end{align}
    Here, terms irrelevant to $\Omega_{\rm r}(t)$ and $\Delta_r(t)$ are omitted. Thus, maximizing ${\rm SNR}_{\rm itn}$ reduces to maximizing $\kappa(\Omega, \Delta)$. The following theorem captures the highest $\kappa(\Omega, \Delta)$ when both the 
frequency detuning and transmission power are large.
\begin{theorem}
\label{theorem1}
    When $\Omega\gg1$ and $|\Delta| \gg \Omega$, the optimal  $\Omega^\star$ that maximizes $\kappa(\Omega, \Delta)$ is given by 
    \begin{align}
        {\Omega^{\star}} = \sqrt{k_1|\Delta|},
    \end{align}
where $k_1\overset{\Delta}{=} \sqrt{\frac{ C_3 
	}{4C_1^2}\left(\sqrt{C_0^2B_1^2 + 16C_1^2} - C_0B_1\right)}$.
\end{theorem}
\begin{proof}
    (See Appendix~\ref{appendix_theorem1}). 
\end{proof}
Using Theorem~\ref{theorem1} and the 
relationship $P_{\rm tx}(t) = \frac{\hbar^2}{2\mu_{34}^2Z_0G_{\rm tx}'{h'}^2}\Omega_{\rm r}^2(t)$,  we are able to directly determine the optimal power that maximizes SNR as:
\begin{align}\label{eq:P_opt1}
	P_{\rm tx}^{\star}(t) = \frac{\hbar^2k_1}{2\mu_{34}^2Z_0G_{\rm 
	tx}'{h'}^2}|\Delta(t)|.
\end{align}
Equation \eqref{eq:P_opt1} indicates that in the off-resonant case where $|\Delta(t)|$ is large, the transmission power $P_{\rm tx}(t)$ should increase linearly with the absolute frequency detuning $|\Delta(t)|$.

Nevertheless, the optimality of \eqref{eq:P_opt1} only holds for off-resonant fields. When the reference signal is nearly resonant ($\Delta(t) \approx 0$), the transmission power computed by \eqref{eq:P_opt1} approaches zero, which would terminate the sensing task and force ${\rm SNR}_{\rm itn} \approx 0$. We provide the following theorem to 
circumvent this issue.
\begin{theorem}\label{theorem2}
    When $\Delta \rightarrow 0$, the optimal  $\Omega^\star$ that maximizes
    $\kappa(\Omega, \Delta)$ is 
\begin{align}\small 
    \Omega^\star = k_2 \overset{\Delta}{=} \sqrt{\frac{ C_2 
    }{4C_1^2}\left(\sqrt{C_0^2B_1^2 
+ 16C_1^2} - 
	C_0B_1\right)}
.\end{align} 
\end{theorem}
\begin{proof}
    (See Appendix~\ref{appendix_theorem2}). 
\end{proof}
Similar to \eqref{eq:P_opt1}, the optimal transmission power for small $|\Delta(t)|$ 
can be directly derived from Theorem~\ref{theorem2} as 
\begin{align}\label{eq:P_opt2}
	P_{\rm tx}^\star(t) =\frac{\hbar^2k_2^2}{2\mu_{34}^2Z_0G_{\rm 
	tx}'{h'}^2}.
\end{align}

To summarize, while Theorem~\ref{theorem1} captures the near-optimal transmission power for large  
detunings, Theorem~\ref{theorem2} characterizes it for 
small detunings. When $\Delta(t) = \frac{k_2^2}{k_1}$, the transmission powers 
calculated by \eqref{eq:P_opt1} and \eqref{eq:P_opt2} become equal. 
Leveraging both theorems, we design the ITN-limited $P$-trajectory as:  
    \begin{align}\label{eq:heuristic}
        P_{\rm tx}(t) = \left\{
        \begin{array}{ll}
        \frac{\hbar^2k_2^2}{2\mu_{34}^2Z_0G_{\rm 
	tx}'{h'}^2}& {\rm if} \: |\Delta(t)|\le 
                \frac{k_2^2}{k_1} \\
                            \frac{\hbar^2k_1|\Delta(t)|}{2\mu_{34}^2Z_0G_{\rm 
	tx}'{h'}^2}& {\rm 
                                    if} \: 
                \frac{k_2^2}{k_1}<|\Delta(t)|\le \frac{B}{2}
        \end{array}
        \right..
    \end{align}
    As the frequency detuning $\Delta(t)$ of LFM waveform sweeps from $-B/2$ to $B/2$, the transmission power $P_{\rm tx}(t) $ follows \eqref{eq:P_opt1} when ${k_2^2}/{k_1}<|\Delta(t)|\le{B}/{2}$, and remains fixed at \eqref{eq:P_opt2} when $|\Delta(t)|\le {k_2^2}/{k_1}$. 
The designed ITN-limited $P$-trajectory is illustrated by the blue line in 
Fig.~\ref{img:PT}(a). This trajectory can consistently achieve a high ${\rm SNR}_{\rm itn}$ across all frequency detunings.

\subsection{Practical $P$-trajectory Design}


We now consider the design of a practical $P$-trajectory  that accounts for both ETN and ITN. 
As illustrated in Fig.~\ref{img:PT}, the presence of ETN attenuates the overall SNR curve, meaning the ITN-limited $P$-trajectory might not be accurate enough to achieve good sensitivity for self-heterodyne sensing. To address this issue, we treat the ITN-limited $P$-trajectory as an initial trajectory and then refine it to maximize the total SNR using numerical optimization methods. This yields a practical $P$-trajectory. 

For computational tractability, we discretize the continuous-time functions, such as $P_{\rm tx}(t)$ and $\varrho(t)$. 
Let the discrete 
time samples be $t_s = \frac{s - 1}{S}T$, $s \in 
\{1,2,\cdots, S\}$, where $S$ is the number of samples. Our design incorporates two key constraints:
\begin{itemize}
    \item \textbf{Total power constraint}: The average transmission power cannot exceed $P_{\rm avg}$: 
    \begin{align}
{\rm (C1)}\quad	\frac{1}{S}\sum_{s = 1}^{S} P_{\rm tx}(t_s) \le P_{\rm avg},
\end{align}
\item \textbf{Non-negativity power constraint}: The transmission power must be non-negative:
\begin{align}
{\rm (C2)}\quad\quad 	P_{\rm tx}(t_s) \ge 0, \:\forall s. 
\end{align}
\end{itemize}

Our objective is to maximize the average received 
SNR across all time samples, leading to the optimization problem:
\begin{align}\label{eq:apc}
	\max_{\{P_{\rm tx}(t_s)\}_{s=1}^S}\frac{1}{S}\sum_{s = 
	1}^{S}\varrho^2(t_s)  \quad{\rm 
	s.t.}\: {\rm (C1)}\& {\rm (C2)}.
\end{align}
Solving \eqref{eq:apc} yields a practical $P$-trajectory that maintains high sensitivity in the presence of both ETN and ITN. Although the objective function in \eqref{eq:apc} is nonconvex, constraints (C1) and (C2) are convex. We can therefore apply the PDS method to solve \eqref{eq:apc} numerically~\cite{nesterov_primal-dual_2009}. The detailed PDS algorithm is presented in {\color{black}Appendix D}. The ITN-limited $P$-trajectory serves as the initial point to avoid convergence to shallow local optima.

The optimized practical $P$-trajectory is plotted as the pink dotted line in Fig.~\ref{img:PT}(b). It slightly raises the power curve compared to the ITN-limited trajectory while remaining within the high-SNR (dark red) region. This ensures near-optimal sensitivity for the RARE across the entire frequency sweep.
\section{Experimental Results}\label{sec:6}
\begin{table}[t]
	\centering
	\caption{Experimental Configuration}
	\label{table:1}
	\begin{tabular}{|c|c|c|c|c|c|c|c|}
		\hline  
		\bf  Parameters&\bf Values\\
		\hline  
		 Energy levels, $\{\ket{1},\ket{2}\}$& $\{6S_{1/2},6P_{3/2}\}$ \\ \hline
         Energy levels, $\{\ket{3},\ket{4}\}$& $\{60D_{5/2},61P_{3/2}\}$ \\ \hline
            Transition dipole moment, $\mu_{34}$ & $2409\:qa_0$ \\ \hline
            Transition dipole moment, $\mu_{12}$ & $2.586\:qa_0$ \\ \hline
            {\color{black} Decay rate, $\gamma_2$} & 
            {\color{black} $2\pi \times 5.2\:{\rm MHz}$} \\ \hline
            Rabi frequency, $\Omega_{\rm p}$ & $2\pi \times 5.8\:{\rm MHz}$ \\ \hline
            Rabi frequencies, $\Omega_{\rm c}$ & $2\pi \times 1\:{\rm MHz}$ \\ \hline
            Wavelength of probe laser, $\lambda_p$ & $852\:{\rm nm}$ \\ \hline
            Wavelength of coupling laser, $\lambda_c$ & $509\:{\rm nm}$ \\ \hline
            Input power of probe laser, $P_{\rm in}$ & $120\:\mu {\rm W}$ \\ \hline
            Length of vapor cell, $d$ & $0.02\:{\rm m}$ \\ \hline
            Density of atoms, $N_0$ & $4.89\times10^{16}\:{\rm m^{-3}}$ \\ \hline
            Antenna gains, $\{ G_{\rm tx}, G_{\rm tx}',G_{\rm rx}\}$ & $\{10, -30, 10\}\:{\rm dBi}$  \\ \hline
            Target range, $L$ & $\mathcal{U}(10^2\:{\rm m}, 10^4\:{\rm m})$  \\ \hline
            Cross-section, $\sigma_S$ & $10\:{\rm dBsm}$  \\ \hline
                    Frequencies,  $\omega_{34}$ and $\omega_0$& $2\pi\times 3.212\:{\rm GHz}$ \\ \hline
		Bandwidth, $B$  & $150\:{\rm MHz}$ \\ \hline
		Symbol duration, $T$ & $1\:{\rm ms}$  \\ \hline
            Transmitter-to-RARE distance & $1\:{\rm m}$  \\ \hline
            Average transmission power, $P_{\rm avg}$ & $1.5\:{\rm W}$ \\\hline
            Load impedance, $R_{\rm T}$ & $2\:{\rm k\Omega}$ \\\hline
            quantum efficiency, $\eta$ & $0.8$ \\ \hline
	\end{tabular}

    \vspace*{0em}
    \begin{tablenotes}
\footnotesize
\item $^*$$a_0 = 52.9{\rm pm}$ denotes the Bohr radius~\cite{SIBALIC2017319}.
\end{tablenotes}
	\vspace*{-1em}
\end{table}

In this section, we present numerical experiments to verify the effectiveness of the proposed self-heterodyne quantum wireless sensing. 

\subsection{Experimental Configurations}
To simulate a practical scenario, we adopt the experimental configurations for atoms and lasers from~\cite{RydNP_Jing2020}. Unless otherwise specified, the key parameters and their default values are listed in Table~\ref{table:1}.
The sensing performance is evaluated using the root mean square error (RMSE) of the estimated delay, defined as ${\rm RMSE} = \sqrt{\mathrm{E}(\tau - \hat{\tau})^2}$. The target range $L$ is randomly sampled from the uniform distribution $\mathcal{U}(10^2\:{\rm m}, 10^4\:{\rm m})$. Each data point in the figures is obtained from 3000 Monte Carlo trials.
We compare three approaches: 1) classical wireless sensing; 2) the proposed self-heterodyne sensing with a fixed transmission power $P_{\rm tx}(t)\equiv P_{\rm avg}$; 3) and the proposed self-heterodyne sensing with the transmission power optimized by solving problem~\eqref{eq:apc}. For each benchmark, we also provide the CRLB to indicate the performance limit. The RMSE curves are shown as solid lines, while their CRLB counterparts are displayed as dashed lines with identical markers.
Note that the heterodyne sensing scheme described in \emph{Example 2} is excluded from our experiments. This is because the considered bandwidth ($B = 150$ MHz) significantly exceeds the typical instantaneous bandwidth of a RARE (less than 10 MHz). Currently, there exists no established model to accurately characterize the received signal or any viable algorithm for range estimation under such heterodyne sensing configurations with such high bandwidth requirements.

\subsection{Experimental Results}

\begin{figure}
	\centering
	\includegraphics[width=3.5in]{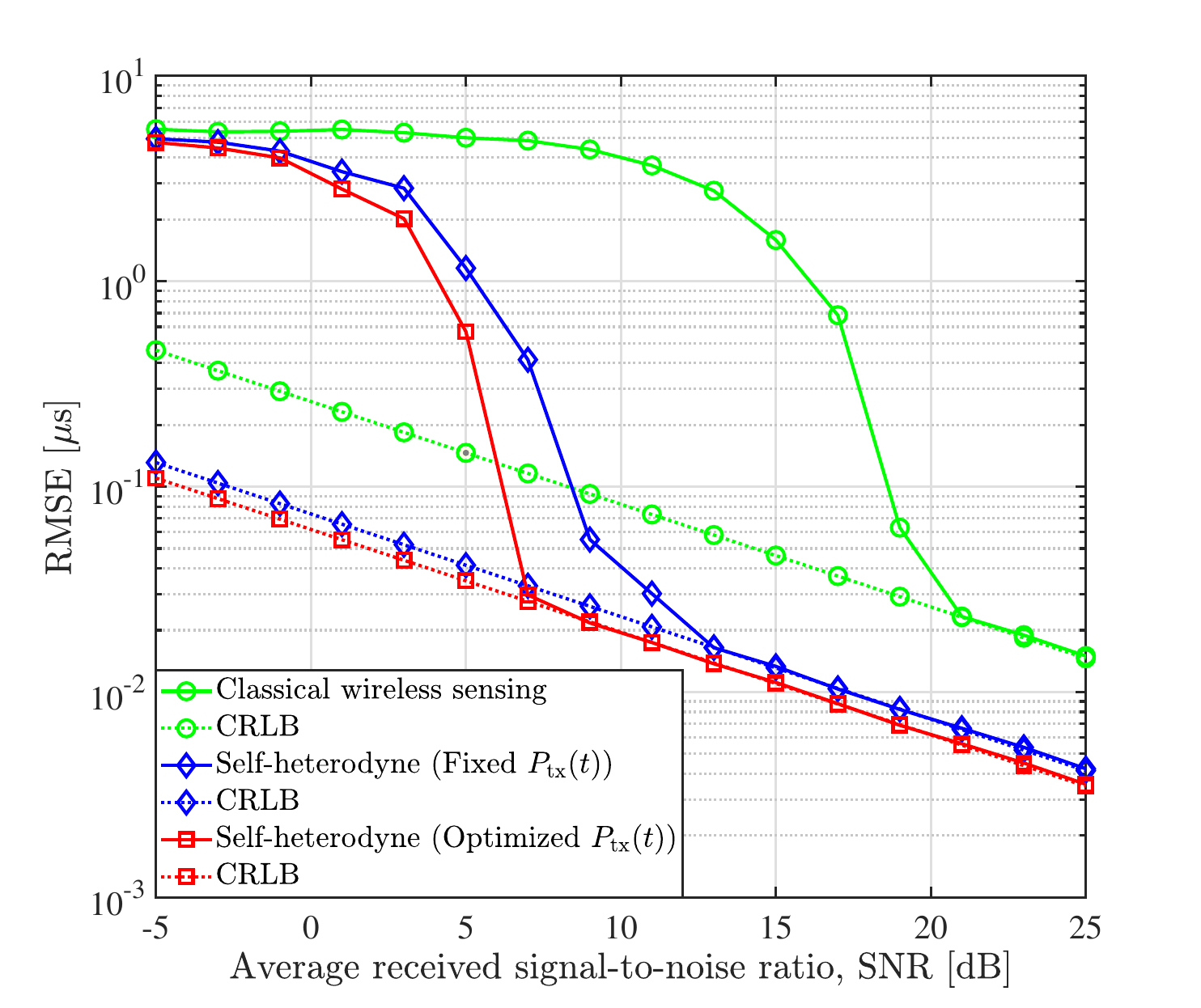}
	\vspace*{-0.8em}
	\caption{RMSE performance versus the average received SNR. 
    } 
	\vspace*{-1em}
	\label{img:receiver_power}
\end{figure}

To begin with, Fig.~\ref{img:receiver_power} shows the RMSE versus the average received SNR, as defined in \eqref{eq:SNR0} for the classical receiver. The SNR variation results from the varying strength of the incident sensing field, ranging from $38\:{\rm nV/m}$ to $1.2\:{\rm \mu V/m}$. 
As expected, the RMSE decreases with increasing SNR for all methods. In particular, the proposed self-heterodyne sensing schemes can reduce the RMSE of classical wireless sensing by 1$\sim$2 orders of magnitude. This substantial improvement stems from the exceptional sensitivity of self-heterodyne RAREs. Furthermore, the optimized 
$P$-trajectory provides an additional 1$\sim$10~dB gain over fixed-power self-heterodyne sensing, validating the effectiveness of power trajectory optimization. Moreover, the derived CRLBs accurately characterize the performance limits of all benchmarks in high-SNR regimes (above 10~dB for self-heterodyne sensing schemes and 20~dB for classical wireless sensing).

\begin{figure}
	\centering
	\includegraphics[width=3.5in]{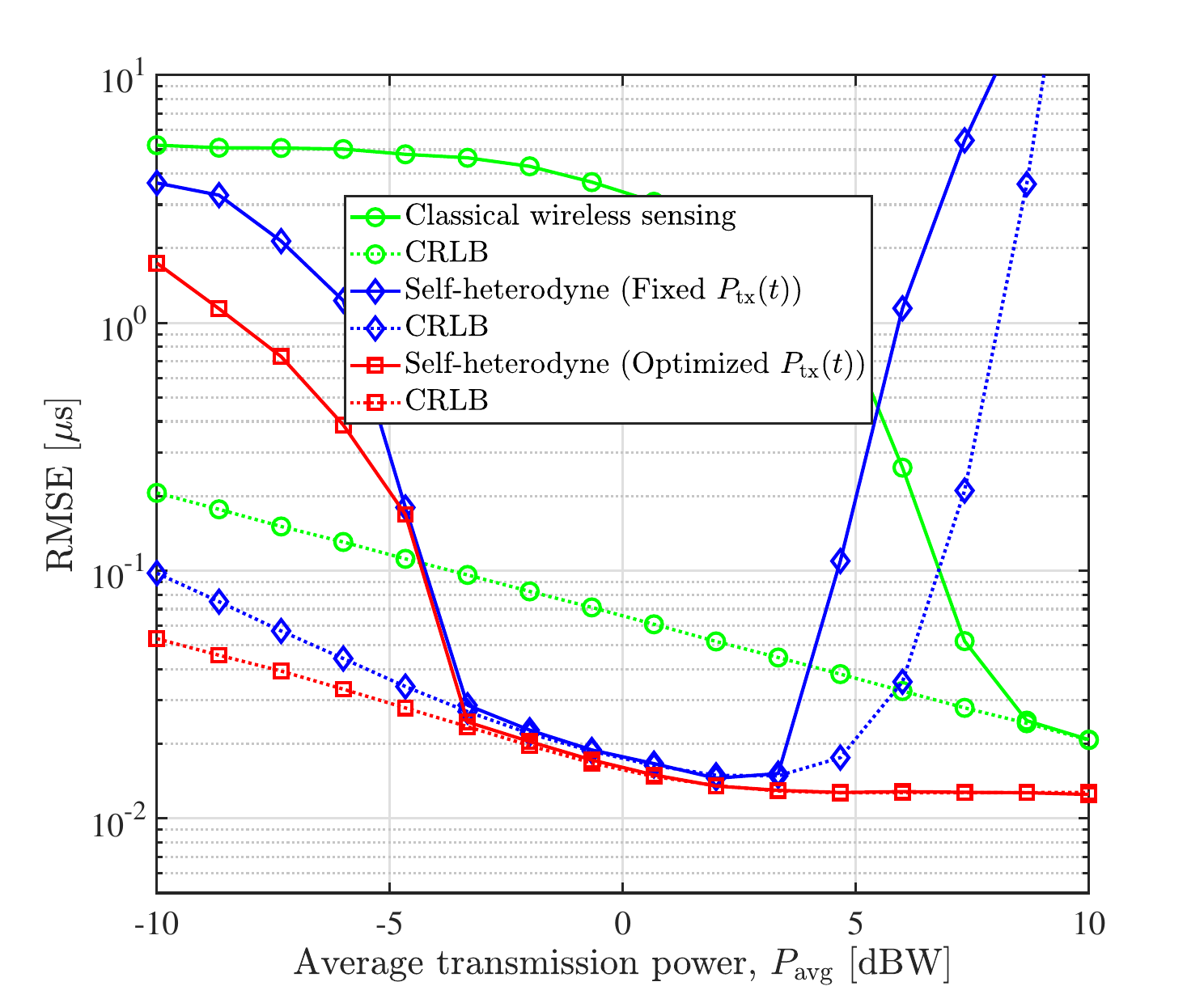}
	\vspace*{-0.8em}
	\caption{RMSE performance versus the average transmission power, $P_{\rm avg}$. } 
	\vspace*{-1em}
	\label{img:MSE_ptx}
\end{figure}

Figure~\ref{img:MSE_ptx} shows the RMSE performance as a function of the average transmission power $P_{\rm avg}$, which ranges from -10 dBW to 10 dBW. 
It is evident from this figure that the proposed self-heterodyne sensing scheme with an optimized power trajectory consumes significantly less transmission power than the classical sensing 
method for achieving the same RMSE level. For instance, to achieve an RMSE of $5\times 10^{-2}\:{\rm \mu s}$, 
the required transmission power is reduced from 2~dBW (classical 
receiver) to -10~dBW (self-heterodyne RARE). In addition, the self-heterodyne sensing with fixed transmission power exhibits a sharp performance degradation when $P_{\rm avg} > 4\:{\rm dBW}$. This deterioration occurs because the fixed power trajectory creates low-SNR regions that become dominant as $P_{\rm avg}$ 
increases, as visible in Fig.~\ref{img:SNR}. In contrast, the optimized power trajectory effectively avoids these low-SNR regions, maintaining accurate range estimation across the entire transmission power range.

\begin{figure}
	\centering
	\includegraphics[width=3.5in]{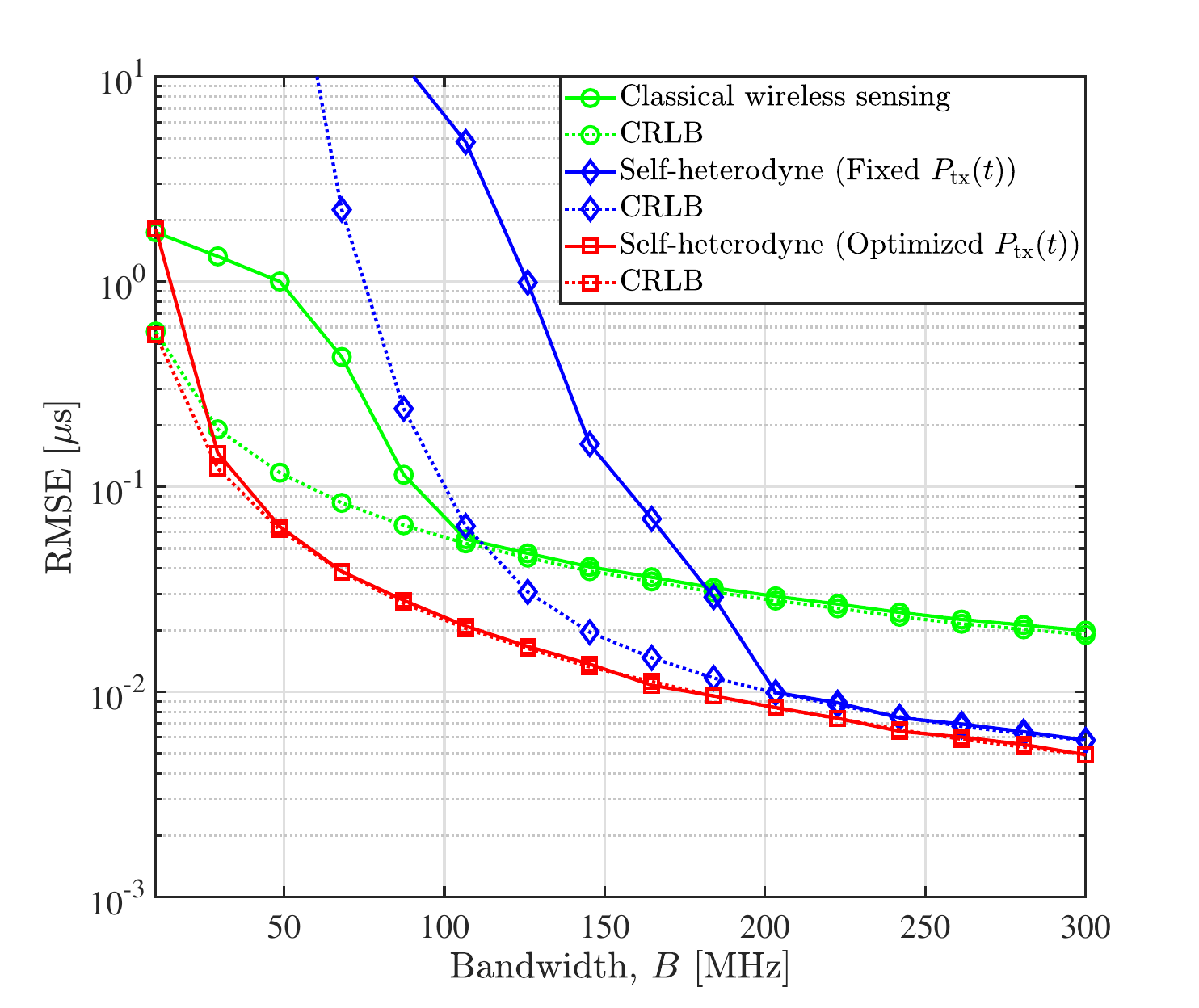}
	\vspace*{-0.8em}
	\caption{RMSE performance versus the system bandwidth, $B$. } 
	\vspace*{-1em}
	\label{img:MSE_band}
\end{figure}

Figure~\ref{img:MSE_band} evaluates the impact of system bandwidth 
$B$ on wireless sensing performance, with $B$ ranging from 10~MHz to 300~MHz.
For large bandwidths (e.g., $B > 200$~MHz), the self-heterodyne sensing schemes can consistently reduce the RMSE of delay estimation from above $2\times 10^{-2}\:{\rm \mu s}$ to below $1\times 10^{-2}\:{\rm \mu s}$. This demonstrates that our self-heterodyne sensing approach can achieve $\sim 100\:{\rm MHz}$-level sensing bandwidth with high sensitivity, substantially surpassing existing heterodyne methods. 
For small bandwidths (e.g., $B < 100$~MHz), self-heterodyne sensing without power trajectory optimization becomes inferior to classical wireless sensing, while the $P$-trajectory optimized version maintains superior performance. This occurs because Rydberg atoms exhibit reduced sensitivity to external fields when the reference field is very strong with small detuning, as evidenced in Fig.~\ref{img:SNR}, highlighting the importance of power-trajectory optimization. Beyond the achieved $\sim 100\:{\rm MHz}$-level sensing bandwidth, it is promising to incorporate lindwidth broadening techniques---such as 
the Zeeman effect and AC Starf shift~\cite{hu_continuously_2022,RydAM_Zhen2019, shi_tunable_2023}---to further increase self-heterodyne sensing bandwidth to the $\sim 1\:{\rm GHz}$ level, paving the way for ultra-broadband quantum wireless sensing. 


\begin{figure}\color{black}
	\centering
	\includegraphics[width=3.5in]{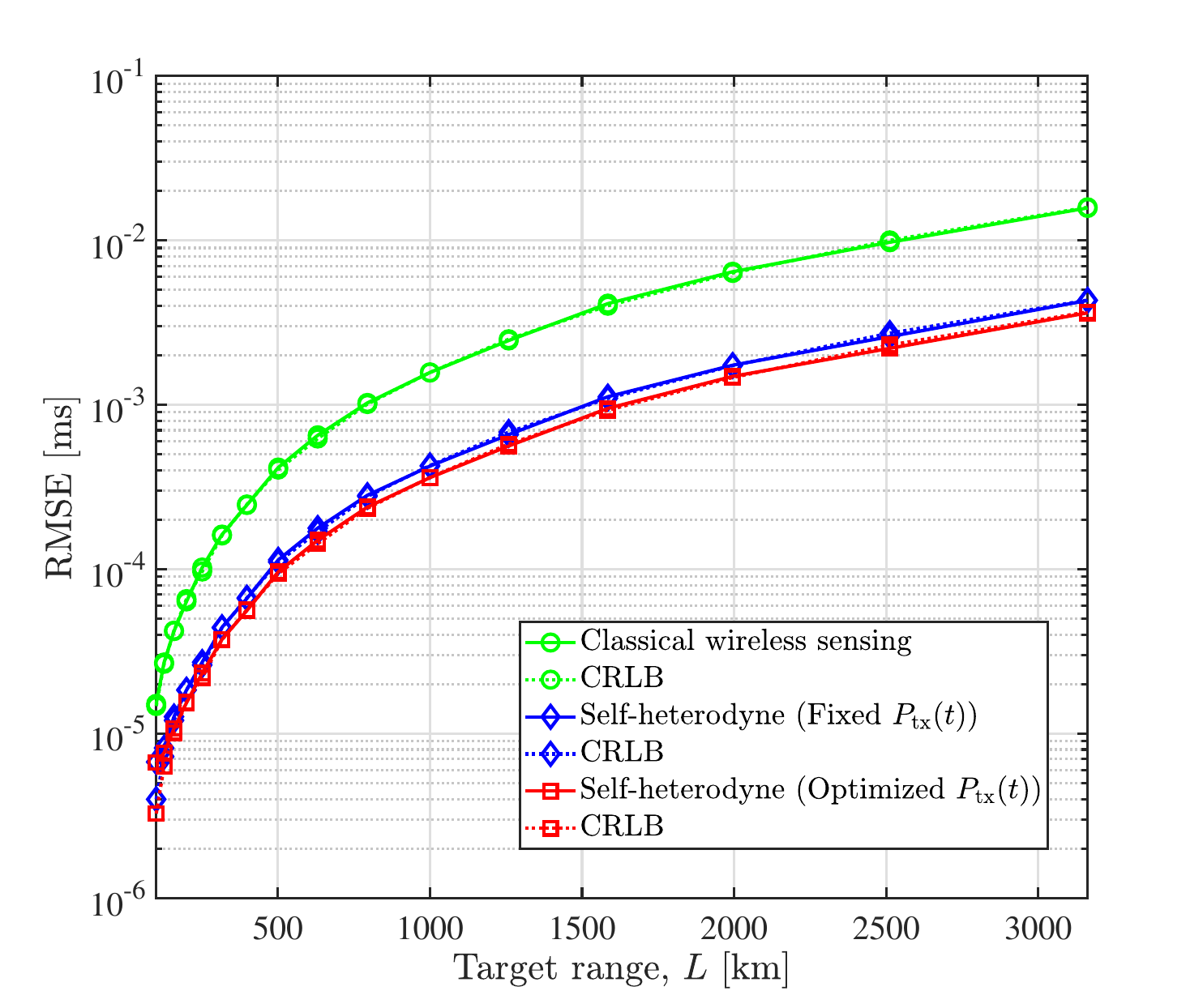}
	\vspace*{-0.8em}
	\caption{RMSE performance versus the target range, $L$. 
    } 
	\vspace*{-1em}
	\label{img:MSE_range}
\end{figure}

{\color{black} Finally, Fig.~\ref{img:MSE_range} demonstrates the capability of the self-heterodyne approaches in long-range target sensing ($L \in [100\:{\rm km}, 3000\:{\rm km}]$). This is a typical scenario in modern radar systems for monitoring distant targets, like aircraft or missiles. 
To address the significant propagation loss, all benchmarks employ enhanced transmitter parameters: an average transmission power of 
$P_{\rm avg} = 30\:{\rm dBW}$, a transmit antenna gain of $G_{\rm tx} = 40~{\rm dBi}$, and a bandwidth of $B = 40\:{\rm MHz}$.  As expected, the RMSE increases with target range for all methods due to growing propagation loss. The proposed self-heterodyne schemes consistently outperform conventional sensing across all ranges, demonstrating the advantage of Rydberg atomic receivers' high sensitivity to weak electromagnetic fields. Notably, at 1,000~km range, the RMSE is reduced from 1.7~$\rm \mu s$ (classical wireless sensing) to approximately 0.4~$\rm \mu s$ (self-heterodyne sensing), confirming the significant advantage of RARE technology for long-range sensing applications.}

\section{Conclusions}\label{sec:7}
We have proposed the self-heterodyne sensing approach to realize quantum 
wireless sensing without extra reference sources. It was discovered that a 
self-heterodyne RARE fundamentally operates as an atomic autocorrelator, 
enabling automatic conversion of the target range into the fluctuation 
frequency of the probe-laser power. Particularly, this autocorrelation 
mechanism allows a RARE to detect sensing signals with a wide bandwidth, 
paving the way for high-resolution quantum wireless sensing. Moreover, we have 
identified the critical role of $P$-trajectory in maximizing RARE 
sensitivity. Numerical results demonstrate that self-heterodyne sensing can 
reduce the MSE of range sensing by orders of magnitude.

This work represents a significant stride towards the development of high-resolution quantum wireless sensing systems. We conclude by outlining potential avenues for future research. First, using self-heterodyne sensing to do multi-target detection remains an open problem.  
Furthermore, it is promising to extend self-heterodyne sensing to RARE-based Integrated Sensing and Communications (RARE-ISAC) systems. In addition, while this work assumes fixed Rabi frequencies and zero frequency detunings for the probe and coupling lasers, dynamically optimizing these time-varying parameters could further enhance RARE sensitivity. {\color{black} Last, it is observed from Figs.~\ref{img:receiver_power}$\sim$\ref{img:MSE_range} that the CRLB is only tight in the high SNR regime.  To obtain a global tight performance bound of self-heterodyne sensing, it is essential to analyze the Ziv-Zakai bound (ZZB) instead~\cite{ZZB_Gu2022, ZZB_Xu2010}, which remains a promising but challenging direction due to the time-varying sensitivity of Rydberg atoms.} 


\appendix{
\subsection{Expression of FIM and Proof of Proposition~\ref{proposition1}}
Denote $\Lambda$ as the loglikehood function $\log p(\bar{y}(t)|\boldsymbol{\kappa})$.  
Each element of the FIM is calculated as
	\begin{small}\begin{align}\label{eq:FIM}
		\left\{
		\begin{array}{l}
			-E\left(\frac{\partial^2 \Lambda}{\partial h^2}\right) =  
			\int_0^T\varrho^2(t)\cos^2(\omega t + \varphi) {\rm d}t \\
			-E\left(\frac{\partial^2 \Lambda}{\partial \omega^2}\right) =  
			h^2\int_0^T\varrho^2(t) t^2 \sin^2(\omega t + \varphi) {\rm d}t \\
			-E\left(\frac{\partial^2 \Lambda}{\partial \varphi^2}\right) =  
			h^2\int_0^T\varrho^2(t)  \sin^2(\omega t + \varphi) {\rm d}t \\
			-E\left(\frac{\partial^2 \Lambda}{\partial \omega\partial 
				\varphi}\right) =  
			h^2\int_0^T\varrho^2(t)  t \sin^2(\omega t + \varphi) {\rm d}t \\
			-E\left(\frac{\partial^2 \Lambda}{\partial h \partial 
			\varphi}\right) =  
			-\frac{h}{2}\int_0^T\varrho^2(t) \sin(2\omega t + 
			2\varphi) 
			{\rm d}t\\
			-E\left(\frac{\partial^2 \Lambda}{\partial h \partial \omega}\right) 
			=  
			-\frac{h}{2}\int_0^T\varrho^2(t)  t\sin(2\omega t + 
			2\varphi) 
			{\rm d}t 
		\end{array}
		\right. .
	\end{align}\end{small}
    
For large $\omega$, the functions $\cos(2\omega t + 2\phi)$ and $\sin(2\omega t 
+ 2\phi)$ will fluctuate much faster than $\varrho^2(t)$,  $\varrho^2(t)t$, 
and  $\varrho^2(t)t^2$. In this context, the integration between a former one 
and a latter one will gradually average to zero. Therefore, the elements of 
FIM in \eqref{eq:FIM} are asymptotically
\begin{align}\small
	\left\{
	\begin{array}{l}
		-E\left(\frac{\partial^2 \Lambda}{\partial h^2}\right) \rightarrow  
		\frac{1}{2}	\int_0^T\varrho^2(t) {\rm d}t \\
		-E\left(\frac{\partial^2 \Lambda}{\partial \omega^2}\right) \rightarrow  
		\frac{h^2}{2}\int_0^T\varrho^2(t) t^2 {\rm d}t \\
		-E\left(\frac{\partial^2 \Lambda}{\partial \varphi^2}\right) 
		\rightarrow  
		\frac{h^2}{2}\int_0^T\varrho^2(t)  {\rm d}t \\
		-E\left(\frac{\partial^2 \Lambda}{\partial \omega\partial 
			\varphi}\right) \rightarrow
		\frac{h^2}{2}\int_0^T\varrho^2(t)  t {\rm d}t \\
		-E\left(\frac{\partial^2 \Lambda}{\partial h \partial \varphi}\right) , 
		-E\left(\frac{\partial^2 \Lambda}{\partial h \partial \omega}\right) 
		\rightarrow 0
	\end{array}
	\right..
\end{align}
We define the notations $\bar{\varrho}_0 \overset{\Delta}{=} 
\int_0^T\varrho^2(t) {\rm d}t$, $\bar{\varrho}_1 \overset{\Delta}{=} 
\int_0^T\varrho^2(t) t {\rm d}t$, and $\bar{\varrho}_2 \overset{\Delta}{=} 
\int_0^T\varrho^2(t) t^2 {\rm d}t$.  The inverse of FIM is 
thus
\begin{align}\small
	\mb{I}^{-1}(\boldsymbol{\kappa}) = \left(\begin{array}{ccc}
		\frac{2}{\bar{\varrho}_0}&0&0\\
		0&\frac{2\bar{\varrho}_0}{h^2(\bar{\varrho}_2\bar{\varrho}_0 - 
		\bar{\varrho}_1^2)}&\frac{-2\bar{\varrho}_1}{h^2(\bar{\varrho}_2\bar{\varrho}_0
		 - 
		\bar{\varrho}_1^2)}\\
		0&\frac{-2\bar{\varrho}_1}{h^2(\bar{\varrho}_2\bar{\varrho}_0 - 
			\bar{\varrho}_1^2)}&\frac{2\bar{\varrho}_2}{h^2(\bar{\varrho}_2\bar{\varrho}_0
			 - 
			\bar{\varrho}_1^2)}
	\end{array}\right).
\end{align}
As a result, the CRLB of the frequency parameter $\omega$ is 
\begin{align}
	{\rm CRLB}(\omega) = 
	\frac{2\bar{\varrho}_0}{h^2(\bar{\varrho}_2\bar{\varrho}_0 - 
		\bar{\varrho}_1^2)}.
\end{align}
The associated CRLB of propagation delay $\tau = \frac{\omega}{\alpha} + \tau'$ is 
\begin{align}
	{\rm CRLB}(\tau) = 
	\frac{2\bar{\varrho}_0}{\alpha^2h^2(\bar{\varrho}_2\bar{\varrho}_0 - 
		\bar{\varrho}_1^2)}.
\end{align}\vspace{-2em}

\subsection{Proof of Theorem 1}\label{appendix_theorem1}
The explicit expression of $\kappa(\Omega, \Delta)$ is written as
    \begin{align}\small\label{eq:appendix1}
   \kappa(\Omega, \Delta) &=  2\sqrt{V_{\rm in}}C_0 B_1\frac{\Omega^4(C_2\Omega^2 + 2C_3\Delta^2)}{(C_1\Omega^4 + C_2\Omega^2 + C_3\Delta^2)^2} \notag\\&\times \exp\left\{
    -\frac{\frac{1}{2}C_0B_1\Omega^4}{C_1\Omega^4 + C_2\Omega^2 + C_3\Delta^2}
    \right\}.
\end{align}
When $\Delta \gg \Omega$ and $\Omega \gg 0$, we have $C_3\Delta^2 \gg C_2\Omega^2$ and $C_1\Omega^4 \gg C_2\Omega^2$. In this way, we can safely remove the term $C_2\Omega^2$ in \eqref{eq:appendix1} and rewrite $\kappa(\Omega, \Delta) $ as 
\begin{align}\small\label{eq:xi}
    \kappa(\Omega, \Delta)  &\propto  e^{
    -\frac{\frac{1}{2}C_0B_1\Omega^4}{C_1\Omega^4 + C_3\Delta^2}
    } \frac{2C_3\Omega^4\Delta^2}{(C_1\Omega^4 + C_3\Delta^2)^2}  \\
    &= e^{
    -\frac{\frac{1}{2}C_0B_1\frac{\Omega^4}{\Delta^2}}{C_1\frac{\Omega^4}{\Delta^2} + C_3}
    } \frac{2C_3\frac{\Omega^4}{\Delta^2}}{\left(C_1\frac{\Omega^4}{\Delta^2} + C_3\right)^2} \overset{\Delta}{=} \xi \left(\frac{\Omega^4}{\Delta^2}\right).  \notag
\end{align}
Thereafter, maximizing $\kappa(\Omega, \Delta) $ is equivalent to maximizing 
$\xi \left(\frac{\Omega^4}{\Delta^2}\right)$. Specifically, the derivative of 
$\xi (x)$ is 
\begin{align}
    \frac{{\rm d}\xi(x)}{{\rm d}x}
    = 2C_3e^{
    -\frac{\frac{1}{2}C_0B_1x}{C_1x + C_3}
    }\frac{ C_3^2-\frac{1}{2}C_0C_3B_1x -C_1^2x^2  }{\left(C_1x + C_3\right)^4}.  
\end{align}
After some tedious calculations, the maximum of $\xi (x)$ is achieved when  
\begin{align}
    x^{\star} = \frac{ C_3 }{4C_1^2}\left(\sqrt{C_0^2B_1^2 + 16C_1^2} - C_0B_1\right).
\end{align}
By invoking the relationship $\frac{{\Omega^{\star}}^4}{{\Delta^{\star}}^2} = x^{\star}$, 
we obtain Theorem~\ref{theorem1}.

\subsection{Proof of Theorem~\ref{theorem2}}\label{appendix_theorem2}
When $\Delta \rightarrow 0$, we can safely remove the term $C_3\Delta^2$ in $\kappa(\Omega, \Delta) $ and simplify it as 
\begin{align}\label{eq:zeta}
    \kappa(\Omega, \Delta)  &\propto e^{
    -\frac{\frac{1}{2}C_0B_1\Omega^2}{C_1\Omega^2 + C_2}
    } \frac{C_2\Omega^2}{(C_1\Omega^2 + C_2)^2} \overset{\Delta}{=} \zeta(\Omega^2).
\end{align}
Thereafter, maximizing $\kappa(\Omega, \Delta) $ is equivalent to maximizing 
$\zeta(\Omega^2)$. Moreover, by comparing the equations \eqref{eq:xi} and 
\eqref{eq:zeta}, one can figure out that the only difference between the 
functions 
$\xi(x)$ and $\zeta(x)$ is that the coefficient $C_3$ in $\xi(x)$ is changed as 
$C_2$ in $\zeta(x)$. 
Therefore, the optimal $x^{\star}$ that maximizes $\zeta(x)$ is clearly 
\begin{align}
    x^{\star} = \frac{ C_2 }{4C_1^2}\left(\sqrt{C_0^2B_1^2 + 16C_1^2} - C_0B_1\right).
\end{align}
By invoking the relationship ${\Omega^{\star}}^2 = x^{\star}$, we obtain 
Theorem~\ref{theorem2}.


\subsection{The PDS method for solving \eqref{eq:apc}}
We consider to solve \eqref{eq:apc} using the PDS method~\cite{nesterov_primal-dual_2009}. Firstly, we introduce 
the Lagrange multipliers $\boldsymbol{\nu} = 
[\nu_1,\cdots,\nu_S]^T$ and $\upsilon$ as well as a positive super-parameter 
$\beta>0$. We further define the optimization variable 
as $\mb{p} = [P_{\rm tx}(t_1), \cdots, P_{\rm tx}(t_S)]^T$.
The augmented Lagrangian function of \eqref{eq:apc} can be written as 
\begin{align}
	\mathcal{L} &= g(\mb{p}) + \upsilon f(\mb{p}) + 
		\boldsymbol{\nu}^T\mb{h}(\mb{p}) +{\beta}(f^2(\mb{p}) + 
		\|\mb{h}(\mb{p})\|^2),
\end{align}
where functions $g(\mb{p})$,   $f(\mb{p})$, and $\mb{h}(\mb{p})$ are 
defined 
as $g(\mb{p}) 
\overset{\Delta}{=} -\frac{1}{S}\sum_{s = 
	1}^{S}\varrho^2(t_s)$, $f(\mb{p}) \overset{\Delta}{=} \left(\sum_{s = 
	1}^{S} 
	P_{\rm tx}(t_s) - S 
P_{\rm avg}\right)^+$, and 
$\mb{h}(\mb{p}) \overset{\Delta}{=} [(-P_{\rm tx}(t_1))^+,\cdots, (-P_{\rm 
tx}(t_S))^+]$. Next, we introduce a superscript $i$ to each variable as the 
iteration index, and the iterative formulas of $\mb{p}^{i}$, 
$\boldsymbol{\mu}^{i}$, 
and $\upsilon^{i}$ are written as 
\begin{align}
	\left(\begin{array}{l}
		\mb{p}^{i + 1}\\
		\upsilon^{i+1}\\
		\boldsymbol{\nu}^{i + 1} 
	\end{array}
	\right) = 
	\left(\begin{array}{l}
		\mb{p}^{i}\\
		\upsilon^{i}\\
		\boldsymbol{\nu}^{i}
	\end{array}
	\right) - \varpi^i\left(\begin{array}{l}
		\partial \mathcal{L}/\partial \mb{p}|_{\mb{p} = \mb{p}^i}\\
		-\partial \mathcal{L}/\partial \upsilon|_{\upsilon = \upsilon^i}\\
		-\partial \mathcal{L}/\partial \boldsymbol{\nu}|_{\boldsymbol{\nu} = 
		\boldsymbol{\nu}^i}
	\end{array}
	\right),
\end{align}
where the step length $\varpi^{i} > 0$ is a sufficiently small value. 
By simultaneously optimizing $\mb{p}$, 
$\boldsymbol{\mu}$, 
and $\upsilon$ until convergence, the $P$-trajectory can be finally obtained.
}

\bibliographystyle{IEEEtran}
\bibliography{Reference.bib}
\end{document}